\documentclass[11pt]{article}
\usepackage{amsmath,amsthm,amssymb,color,verbatim,graphicx}
\newcommand{\remove}[1]{}
\newtheorem{thm}{Theorem}[section]
\newtheorem{claim}[thm]{Claim}
\newtheorem{lem}[thm]{Lemma}
\newtheorem{define}[thm]{Definition}
\newtheorem{cor}[thm]{Corollary}

\newtheorem{conjecture}{Conjecture}
\newtheorem{THM}{Theorem}
\newtheorem{COR}[THM]{Corollary}

\newtheorem{remark}[thm]{Remark}

\renewcommand{\remove}[1]{}

\newcommand{\poly}{{\rm poly}}

\newcommand{\eps}{{\varepsilon}}
\renewcommand{\l}{\left}
\renewcommand{\r}{\right}

\newcommand{\de}{{\delta}}

\newcommand{\comments}[1]{}
\newcommand{\rank}{\textnormal{rank}}

\newcommand{\hinf}{H_\infty}
\newcommand{\set}[1]{\{ #1 \}}
\newcommand{\zo}{\{0,1\}}
\newcommand{\supp}{{\rm supp}}
\newcommand{\ext}{\textnormal{Ext}}

\def\F{{\mathbb{F}}}
\def\B{\set{0,1}}

\newcommand{\Z}{\mathbb{Z}}
\newcommand{\Zq}{\mathbb{Z}^*_q}
\newcommand{\Zp}{\mathbb{Z}_p}

\newcommand{\N}{\mathbb{N}}
\newcommand{\E}{\mathbb{E}}
\newcommand{\Fp}{\mathbb{F}_p}
\newcommand{\Fq}{\mathbb{F}_q}

\newcommand{\C}{\mathcal{C}}

\newcommand{\Sym}{\mathrm{Sym}}
\newcommand{\Bo}{\mathbf{Bohr}}

\renewcommand{\Pr}{\mathbf{Pr}}
\newcommand{\Fn}{ {\mathbb{F}}_p^n}

\def\draft{0}   

\ifnum\draft=1 
    \def\ShowAuthNotes{1}
\else
    \def\ShowAuthNotes{0}
\fi

\ifnum\ShowAuthNotes=0
\newcommand{\authnote}[2]{{ \footnotesize \bf{\color{red}[#1's Note: {\color{blue}#2}]}}}
\else
\newcommand{\authnote}[2]{}
\fi

\newcommand{\AuthornoteA}[2]{{\sf\small\color{green}{[#1: #2]}}}

\newcommand{\Anote}{\AuthornoteA{A}}

\begin{document}
\title{Deterministic Extractors for Additive Sources}

\author{
Abhishek Bhowmick\thanks{bhowmick@cs.utexas.edu, Department of Computer Science, The University of Texas at Austin. Research supported in part by NSF Grants CCF-0916160 and CCF-1218723.}
\and
Ariel Gabizon\thanks{ariel.gabizon@gmail.com, Computer Science Department, Technion, Haifa, Israel. The research leading to these results has received funding from the European Community's Seventh Framework Programme (FP7/2007-2013) under grant agreement number 257575.}
\and
Th\'ai Ho\`ang L\^e\thanks{leth@math.utexas.edu, Department of Mathematics, The University of Texas at Austin.}
\and
David Zuckerman\thanks{diz@cs.utexas.edu, Department of Computer Science, The University of Texas at Austin. Research supported in part by NSF Grants CCF-0916160 and CCF-1218723.}
}

\maketitle
\thispagestyle{empty}
\begin{abstract}
We propose a new model of a weakly random source that admits randomness extraction.
Our model of additive sources includes such natural sources as uniform distributions on
arithmetic progressions (APs), generalized arithmetic progressions (GAPs), and Bohr sets, each of which
generalizes affine sources.
We give an explicit extractor for additive sources with linear min-entropy over both $\Z_p$ and $\Z_p^n$, for large prime $p$, although our
results over $\Z_p^n$ require that the source further satisfy a list-decodability condition.
As a corollary, we obtain explicit extractors for APs, GAPs, and Bohr sources with linear min-entropy, although
again our results over $\Z_p^n$
require the list-decodability condition.

We further explore special cases of additive sources.  We improve previous constructions of line sources (affine sources of dimension 1),
requiring a field of size linear in $n$, rather than $\Omega(n^2)$ by Gabizon and Raz.  This beats the non-explicit bound of
$\Theta(n \log n)$ obtained by the probabilistic method.  We then generalize this result to APs and GAPs.
\end{abstract}


\newpage
\setcounter{page}{1}

\section{Introduction}
High-quality randomness is needed for a variety of applications. However, most physical sources are only weakly random.  Moreover, such weak sources arise in cryptography when an adversary learns information about a uniformly random string.  It is therefore natural and important to try to extract the usable randomness from a weak source.  It is impossible to extract even one bit of randomness from a natural yet large enough class of sources using a single function \cite{sv}. There are two ways to counter this. One is to extract with the help of a small amount of randomness; this is called a seeded extractor \cite{nz}. Our focus is on the second way: to extract only from more structured sources (and not allow any auxiliary randomness). Such a function is called a \emph{deterministic} (or seedless) extractor.

We now give a formal definition of extractors.
In the following definition the term \emph{source} simply refers to a random variable.
\begin{define}
A function $\ext: \set{0,1}^n \rightarrow \set{0,1}^m$ is an \emph{$\epsilon$-extractor} for a family of sources $\mathcal{X}$
if for every $X \in \mathcal{X}$,
the distribution $\ext(X)$ is $\eps$-close in statistical (variation) distance to $U_m$.
Here $U_m$ denotes the uniform distribution on $m$ bits.
\end{define}

We measure the randomness in a source $X$ using min-entropy.

\begin{define}
The \emph{min-entropy} of a random variable~$X$ is
\[ \hinf(X)=\min_{x \in \supp(X)}\log_2(1/\Pr[X=x]).\]
If $X \subseteq \{0,1\}^n$, we say that $X$ has \emph{entropy rate} $\hinf(X)/n$.
\end{define}

The probabilistic method shows that if $|\mathcal{X}| \leq 2^{2^{.9k}}$,
where $k$ is the min-entropy of each source, then there exists a deterministic extractor for $\mathcal{X}$.
Constructing such an extractor explicitly is a much harder challenge.
One type of source for which deterministic extractors have been constructed is an \emph{affine source} -
a uniform distribution over an affine subspace of a vector space ~\cite{GabR,Bou:affine,deVG,yeh:affine,l:affine,BourDL14}.
  In this paper, we explore generalizations of affine sources with more minimal structure.
  We show that an explicit deterministic extractor can be constructed for a broad generalization of
  affine sources that we call \emph{additive sources}.
\begin{remark}
Throughout the paper we often abuse notation and 
refer to a \emph{set} $X$ as a source.
The source is actually the random variable uniformly distributed on 
the set $X$.
\end{remark}

Before presenting our general notion of an additive source, it will be instructive
to look at two simpler natural generalizations of affine sources that are special cases
of our notion.
The first generalizes an affine source when viewed as the \emph{image} of a linear map.
The second generalizes a linear source when viewed as the \emph{kernel} of a linear map.

\paragraph{Generalized arithmetic progressions}
An affine subspace in $\Z_p^n$ may be viewed as the set of elements
$V=\set{a_1\cdot t_1+\ldots a_r\cdot t_r + b|  t_1,\ldots,t_r \in \Z_p}$
for some fixed $a_1,\ldots,a_r,b\in \Z_p^n$ such that $a_1,\ldots,a_r$ are linearly independent.
One relaxation of this definition would be to not insist that $a_1,\ldots,a_r$ be linearly independent,
and allow the $t_i$'s to only range through a subset of $\Z_p$ of the form $\set{0,\ldots,s-1}$, rather than all of $\Z_p$.
The result is exactly what is known as a \emph{generalized arithmetic progression} (GAP).
That is, a $(r,s)$-GAP is a set of the form
\[A=\set{a_1\cdot t_1+\ldots a_r\cdot t_r + b|  0\leq t_1,\ldots,t_r \leq s-1}\]
for some fixed $a_1,\ldots,a_r,b\in \Z_p^n$ and $s\leq p$.
\paragraph{Bohr sets}
A linear subspace in $\Z_p^n$ may also be viewed as the set of elements
$v\in \Zp^n$ such that for all $i=1,\ldots,d$, $L_i(v) =0$,
for some fixed linear functions $L_1,\ldots,L_d:\Zp^n\to \Zp$.
A relaxation of this definition could be to look at the set of elements $v\in V$
such that $L_i(v)$ is `close to zero' for every $i\in [d]$.
We could define the \emph{distance from zero} of an element $a\in\Z_p$ by looking
at $a$ as an integer in $\set{0,\ldots,p-1}$, and taking the minimum of the distances
between $|p-a|$ and $|a-0| = |a|$.
Equivalently, we could define it as $\| a/p \|$ where $\| \cdot \|$ denotes the
distance to the nearest integer.
The resulting definition is what is known as a $\emph{Bohr set}$ in $\Z_p^n$.
That is, a $(d,\rho)$-Bohr set is a set $B$ of the form
\[B=\set{v\in \Zp^n: \| L_i(v)/p \|<\rho}\]
for some fixed $0<\rho<1$ and linear functions $L_1,\ldots,L_d:\Zp^n\to \Zp$.

In fact, as opposed to subspaces GAPs and Bohr sets can be defined not
just in $\Zp^n$ but in any abelian group.
See Definitions \ref{dfn:gap} and \ref{dfn:bohr} for the definitions in a general abelian group.

\noindent We proceed to describe our general notion of an additive source.
\subsection{Defining additive sources}
Before defining an additive source, we give some intuition on the definition we chose
and the pitfalls of other natural definitions.

We work in an abelian group $G$, which is usually $\Z_p$ or $\Z_p^n$ under addition.
A first attempt at a minimal structure that generalizes subspaces is to require $X$ to have small doubling:  $|X+X| \leq C|X|$ for small $C>1$.
(Here $A+B$ denotes the set $\set{a+b|a \in A,b \in B}$.)
The Cauchy-Davenport Theorem implies that for $A\subseteq \Z_p$, $|A+A| \geq \min \{2|A|-1,p\}$.
Kneser's theorem, which extends the Cauchy-Davenport theorem, implies the same is true for any $A\subseteq \Z_p^n$ that is not contained in a strict subgroup of $\Z_p^n$.
So, for obtaining a large class of sources it makes sense to look at $C\geq 2$.
However, even for $C=2$ we get a class of sources for which deterministic extraction is impossible:  For let $f:G \rightarrow \zo$ be any such purported extractor.
Then the uniform distribution on the larger of $f^{-1}(0)$ and $f^{-1}(1)$ gives a counterexample.  If this seems artificial and one asks about smaller sets,
we could start with any $B$ such that $|B+B| \leq 2|B|$, such as an arithmetic progression, and then the larger of
$f^{-1}(0) \cap B$ or $f^{-1}(1) \cap B$ gives a counterexample for $C=4$.
A similar attempt at a definition would be to lower bound the \emph{additive energy} - a quantity that measures how many sums in $X+X$ lead to the same value. However, this is also insufficient, as sets with small doubling have large energy.

In light of the above, we seek to impose an additional condition, besides small doubling.
This extra condition involves the notion of symmetry sets from additive combinatorics.
A symmetry set for a set $X \subseteq G$ with parameter $\gamma>0$ is defined as
\[\Sym_{\gamma}(X)=\{g \in G: |X \cap (X+g)|\geq \gamma  |X|\}.\]
In other words, an element is in $\Sym_{\gamma}(X)$ if it can be expressed as $x-x'$, for $x,x' \in X$, in at least $\gamma |X|$ ways.
We shall be interested in the setting where $\gamma$ is close to $1$.
The simplest examples of sets with large symmetry sets are subgroups and cosets of subgroups.
Specifically, if $X$ is a subgroup or a coset of a subgroup, then $\Sym_1(X) =X$.

We note that large symmetry sets don't imply small doubling.
For example, if we start with a set~$Y$ with $\Sym_{1-\alpha}(Y)$ large, then we could choose a set $T$ of size $2\alpha |Y|$ with large doubling,
such as a Sidon set or random set, and set $X=Y \cup T$.  Then $\Sym_{1-3\alpha}(X)$ is large but $X$ has large doubling.
Yet this counterexample isn't completely satisfactory, because for extraction it would suffice that whenever $X$ has $\Sym_{1-\alpha}(X)$ large,
there exists a large $X' \subseteq X$, $|X'| \geq (1-\eps)|X|$, where $X'$ has small doubling.  We also give a counterexample to this weakened question.
To give a counterexample with $p^{1/d}$ large symmetry sets, it's easiest to work in $\Z_p^d$.  Pick a large-doubling set $T$ in $\Z_p^{d-1}$,
and let $X=T \times \Z_p$.  Then $\Sym_1(X)$ contains $\Z_p$, but $X+X$ has size $\Theta(|T|^2 p) = \Theta(|X|^2 /p)$.
The same is true for large subsets of $X$.
If we worked in $\Z_p$ instead, we could take a union of intervals, which would give slightly weaker parameters.

Thus, we define an additive source to be (the uniform distribution on) a set that has small doubling and has a large symmetry set.

\begin{define}[Additive source] A set $X$ in a finite abelian group $(G,+)$ is called an $(\alpha, \beta, \tau)$-additive source if $|X+X|\leq |X|^{1+\tau}$ and \[|\Sym_{1-\alpha}(X)|\geq |X|^{\beta}.\]
\end{define}

\noindent In Sections \ref{sec:extrzp} and \ref{sec:extrzpn} we show
that GAPs and Bohr sets in $\Zp$ and $\Zp^n$ are indeed captured by our definition of additive sources.
\noindent
As far as we know, these have not been studied in the extractor literature.

One can easily see that there are doubly exponentially many $(\alpha, \beta, \tau)$-additive sources for reasonably small $\alpha, \tau$ and any constant $\beta<1$. Due to this, there is no succinct representation of a general additive source, unlike affine sources. The only other natural family with doubly exponentially many sources is the family of independent sources.

\subsection{Related work}
We review relevant previous work.
The first class of additive sources considered were bit fixing sources by Chor et al \cite{cfghrs}, and then by Kamp and Zuckerman \cite{KZ} and Gabizon, Raz and Shaltiel in \cite{GabRS}. Next, in the more general case of affine sources, Bourgain obtained extractors for constant entropy rate \cite{Bou:affine} over $\F_2$, with improvements to slightly subconstant rate by
Yehudayoff \cite{yeh:affine} and Li \cite{l:affine}. In the case of large fields, extractors for affine sources were given by Gabizon and Raz \cite{GabR} and more recently by Bourgain, Dvir and Leeman \cite{BourDL14}.  DeVos and Gabizon \cite{deVG} gave constructions interpolating between these extreme cases. Generalizations of affine sources have also been studied in the work of Dvir, Gabizon and Wigderson \cite{dgw} and Ben-Sasson and Gabizon \cite{BenG} where the authors look at polynomial sources and by Dvir \cite{Dvir:variety} where varieties are considered. Special cases of affine sources have also been studied by Rao \cite{Rao08}. Gabizon and Shaltiel \cite{GabShal} constructed a weaker object called a disperser over large fields for affine sources. Ben-Sasson and Kopparty \cite{B-SKop12} constructed dispersers for affine sources with min-entropy $6n^{4/5}$ over $\F_2$.
Shaltiel \cite{ShaltielDisp11} improved on this and constructed a disperser for min-entropy $n^{o(1)}$ over $\F_2$.

\subsection{Our results}
In our main theorem for $\Z_p$ for $p$ a large prime, we construct an extractor for additive sources for any constant entropy rate.
More specifically, our construction works whenever $p$ is large enough, and for $(\alpha, \beta, \tau)$-additive sources whenever $\alpha$ and $\tau$ are small enough. Specifically, for $p$ an $n$-bit prime, we need $\alpha < 1/n$, and we extract about $\log(1/(\alpha n))$ bits. Thus, $\alpha=1/\poly(n)$ leads to $\Omega(\log n)$ random bits whereas $\alpha=1/p^{\gamma}$ leads to $\Omega(n)$ random bits.
We now state our main theorem over $\Z_p$.
\begin{THM}\label{THM:zp}For every $\de, \beta>0$, there exists $\tau>0$ and $p_0$ such that for all primes $p>p_0$ and $\alpha>0$, the following holds. There is an explicit efficient $\eps$-extractor $\ext :\Z_p \rightarrow \{0,1\}^m$, for ($\alpha, \beta, \tau$)-additive sources of entropy rate $\de$ in $\Z_p$ where $\eps=\l(3\alpha+p^{-\Omega_{\beta,\de}(1)}\r) 2^{m/2} \log p$.
\end{THM}

As a corollary, we obtain extractors for GAPs for any constant entropy rate.
\begin{COR}[GAP Sources] \label{cor:zpgap} For all $\de>0$, there exists $c,p_0$ such that for all primes $p \geq p_0$ the following holds. For all integers $r \geq c$, and all primes $p \geq c^{r/\de}$ the following holds. There exists an explicit efficient $\eps$-extractor $\ext:\Z_p \rightarrow \{0,1\}^m$, for ($r,p^{\de/r}$)-GAP sources (of entropy rate $\de$) in $\Z_p$ where $\eps=\l(\frac{3r}{p^{0.9\de/r}}+p^{-1/2}\r) 2^{m/2} \log p$.
\end{COR}

Observe that the only restriction we put is that $r \geq c$ and $p \geq c^{r/\de}$ which simply means that the side lengths of the GAP (that is, $p^{\de/r}$) have to be larger than some fixed constant $c$ and the dimension has to be larger than a fixed constant $c$. Thus, if we let $r$ be a constant, then we can extract a constant fraction of the min entropy, that is $\Omega(\de\log p)$ bits.

As another corollary, we obtain extractors for Bohr sources. We state it for constant $\rho$ for simplicity. It can be easily generalized to any arbitrary $\rho$. 
\begin{COR}[Bohr Sources]\label{cor:zpbohr} Let $\rho, \alpha>0$ and $S \subseteq \Z_p$ with $|S|=d$ be arbitrary. Then for prime $p=\Omega\l(\l(\frac{d}{\alpha}\r)^d\r)$, there exists an explicit efficient $\eps$-extractor $\ext:\Z_p \rightarrow \{0,1\}^m$, for ($d, \rho$)-Bohr sources of entropy rate $\de$ in $\Z_p$ where $\eps=\l(3\alpha+p^{-\Omega(1)}\r) 2^{m/2} \log p$.

\end{COR}

Next, we construct an extractor for additive sources in $\Z_p^n$ for large enough ~$p$ (polynomial in $n$) and any constant entropy rate, provided the source is sufficiently structured additively
and satisfies a certain list decodability property.
As a corollary, we give an extractor for GAPs in $\Z_p^n$ and Bohr sets with constant entropy rate, provided they
satisfy the list decodability property. We note that our extractor works for affine sources even though they may not satisfy the list decodability property. See Section~\ref{sub:affine}.

Our extractors for GAPs and Bohr sets over $\Z_p$, and for Bohr sets over $\Z_p^n$, extract a linear number of bits with exponentially small error.
We also show that all large sets (min-entropy rate close to 1), most sets ($\de$ min-entropy rate for any $\de>0$) and most affine sources ($\de$ min-entropy rate for any $\de>0$) satisfy the list decodability condition. See Remark~\ref{rem:zpn1}, \ref{rem:zpn2} and \ref{rem:zpn3}.


In the final two sections, we study special cases of additive sources. First, we give an extractor for one dimensional affine sources (lines) in $\F_q^n$ which requires only that $q>n$.  This improves the results of Gabizon and Raz \cite{GabR}, which required $q=\Omega(n^2)$.  Surprisingly, it even improves the non-explicit bound obtained via the probabilistic method of $q=\Omega(n\log n)$. 

\begin{THM}[Extractors for lines]\label{thm:normlines} There is an explicit efficient $\eps$-extractor $\ext:\F_q^n \rightarrow \{0,1\}$ for all line sources in  $\F_q^n$ where $\eps \leq 4(n/q)^{1/2}$. 
\end{THM}

We then show the same extractor in fact works for `partial lines' - i.e., arithmetic progressions
in $\Zp^n$.

\begin{THM}[Extractors for APs]\label{THM:normap} There is an explicit efficient $\eps$-extractor $\ext :\F_p^n \rightarrow \{0,1\}^m$ for all $k$-AP sources in  $\F_p^n$ where $\eps \leq 16 \log^2 p \sqrt{np}2^{m/2}/k$.
\end{THM}

Therefore, if we have $k=p^{1/2+\de}$ and $n<p^{\de}$, then we can extract $\de/2 \log p$ bits. Moreover, we show that the general framework of \cite{GabR} for constructing extractors
for affine sources can be generalized to work for GAPs. See Theorem~\ref{THM:normgap}.
As a corollary, we extend a result of DeVos and Gabizon \cite{deVG}
to obtain extractors for GAPs in $\Zp^n$.

\subsection{Techniques and Proof Overview}

\subsubsection{Extractors for additive sources in $\Z_p$.}

For our proofs it will be convenient to define the notion of a
\emph{multiplicative source}.
The definition simply corresponds to that of an additive source
with multiplicative notation.
Formally,
\begin{define}[Multiplicative source]
Fix positive constants $0<\alpha,\beta,\tau\leq 1$.
Let $Y$ be a subset of  a finite abelian group $(G,\cdot)$.
 We define the set $\Sym_{1-\alpha}(Y)\subseteq G$ by
\[ \Sym_{1-\alpha}(Y)\triangleq \{g \in G: |Y \cap (Y\cdot g)|\geq (1-\alpha)\cdot |Y|\}.\]
We say that $Y$ is an $(\alpha, \beta, \tau)$-multiplicative source if
$|Y\cdot Y|\leq |Y|^{1+\tau}$ and $|\Sym_{1-\alpha}(Y)|\geq |Y|^{\beta}$.

\end{define}

Suppose $X$ is an $(\alpha,\beta,\tau)$-additive source in $\Z_p$.
Our extractor construction is as follows. We describe the construction in detail only for this class of sources. Let $q$ be a large prime, $q=1 \pmod p$ and $g$ be a generator of $\Zq$ of order $p$. Define
$\ext(x)\triangleq \sigma(g^x)$, where  $\sigma:\Z_q \rightarrow \Z_m$ is the function from Lemma \ref{lem:addlem}. Then, it is enough by Lemma~\ref{lem:addlem}, to show that $\l|\E_X e_p(a.g^X)\r|$ for all $a \neq 0$ is small.
 The analysis break down into two main steps:
\paragraph{Step 1: `Encoding' $X$ into a multiplicative source.}
As noted before, we fix a prime $q>p$ such that $q=1 \pmod p$.
For such $q$ there exists an element $g\in \Zq$ of order $p$.
Fix such an element $g$ and look at the map from $\Zp$ to $\Zq$ taking
$x$ to $g^x$.
Let $Y\subseteq \Zq$ be the image of $X$ under this map.
That is,
$Y\triangleq \set{g^x|x\in X}.$
As the subgroup generated by $g$ in $\Zq$ is isomorphic to
$\Zp$ we can show that $Y$ is an $(\alpha,\beta',\tau)$-multiplicative source
in $\Zq$, where $\beta'\sim \beta$.
\paragraph{Step 2: Applying a character sum bound of Bourgain together with an `average to worst-case reduction'.}
The advantage of the transition to a multiplicative source
comes from a theorem of Bourgain that roughly says the following.
Suppose $Y$ is a subset of $\Zq$ such that $|Y\cdot Y|\leq |Y|^{1+\tau}$ for appropriate
$0<\tau<1$. Then, for most $a\in \Z_q$ the sum
\[\widehat{Y}(a)\triangleq \sum_{y\in Y} e_q(a\cdot y)\]
is small in absolute value.
See Theorem \ref{thm:bour} for a precise statement (The theorem does not directly correspond to the description here,
and actually talks about the `$t$'th moment of additive characters over $Y$'.)
If we knew that $|\widehat{Y}(a)|$ is small for \emph{all} $a\in \Zq$
rather than most $a\in \Zq$, we could extract randomness from $Y$ using the XOR lemma (Lemma \ref{lem:addlem}).
Our main insight is that when $\Sym_{1-\alpha}(Y)$ is large, we
can indeed deduce that $|\widehat{Y}(a)|$ is small for all $a\in \Zq$.
We sketch why this is the case.
Assume for contradiction that there is some $a\in \Zq$
such that
\[|\widehat{Y}(a)|= \left|\sum_{y\in Y} e_q(a\cdot y)\right|\]
 is large.
Fix any $a'\in \Sym_{1-\alpha}(Y)$.
As $|Y\cap a'\cdot Y|\geq (1-\alpha)\cdot |Y|$
and each summand is one in absolute value,
the above sum will not change much if we sum over
$a'\cdot Y$ rather than $Y$.
That is, the sum
\[\sum_{y\in a'\cdot Y} e_q(a\cdot y)\]
must also be large in absolute value.
But this sum is equal to
\[\sum_{y\in  Y} e_q(a'\cdot a\cdot y) = \widehat{Y}(a'\cdot a).\]
Thus, $|\widehat{Y}(a'\cdot a)|$ is large for all $a'\in \Sym_{1-\alpha}(Y)$ -
a contradiction as we know that $|\widehat{Y}(b)|$ is small for most $b\in \Zq$.

Thus, for all $a \neq 0$, $\l|\E_X e_p(a.g^X)\r|$ is small. In summary, the extractor construction is
$\ext(x)\triangleq \sigma(g^x)$, where  $\sigma:\Z_q \rightarrow \Z_m$ is the function from Lemma \ref{lem:addlem}.
See Section \ref{sec:extrzp} for full details.

\subsubsection{Extractors for additive sources in $\Z_p^n$.}
Our construction over $\Z_p^n$ follows similar lines but is more involved.
We give a sketch describing the same basic two steps.
Let $X$ be an $(\alpha,\beta,\tau)$-additive source in $\Z_p^n$.
\paragraph{Step 1: `Encoding' $X$ into a multiplicative source.}
  We choose $n$ distinct primes $q_1,\ldots,q_n$ such that for all $i\in [n]$,  $q_i=1 \pmod p$. Let $g_i$ be an element of order $p$ in $\Z_{q_i}^*$.
  Let $q=q_1\cdots q_n$, and let $CRT:\prod_i  \Z_{q_i} \rightarrow \Z_q$ be the `Chinese remaindering map'.

  We look at the map from $\Zp^n$ to $\Z_q$ taking $(x_1,\ldots,x_n)$ to $CRT(g_1^{x_1},\ldots,g_n^{x_n})$.
  Let $Y$ be the image of $X$ under this map.
That is,
$Y\triangleq \set{CRT(g_1^{x_1},\ldots,g_n^{x_n})|(x_1,\ldots,x_n)\in X}.$

We can show that $Y$ is an $(\alpha,\beta',\tau)$-multiplicative source
in $\Zq$,\footnote{Observe that since the vector $(g_1^{x_1},\ldots,g_n^{x_n})$ is non-zero
in all coordinates, the element $CRT(g_1^{x_1},\ldots,g_n^{x_n})$ of $\Z_q$ is indeed in $\Z_q^*$.} where $\beta'\sim \beta$ assuming $q=p^{O(1)}$.
We show that we can indeed get $q=p^{O(1)}$ by observing that
the proof of Linnik's Theorem implies that for large enough $p$,
we can always find appropriate $q_1,\ldots,q_n$ that are all at most
$p^{O(1)}$.

\paragraph{Step 2: Applying a character sum bound of Bourgain together with an `average to worst-case reduction'.}
As in the case of $\Z_p$, we would now like to apply a theorem saying that for $Y\subseteq \Zq$ such that $|Y\cdot Y|\leq |Y|^{1+\tau}$,  $|\widehat{Y}(a)|$ is small for most $0\neq a\in \Z_q$.
The difference from the case of $\Z_p$ is that now we are dealing with a \emph{composite} $q$.
Bourgain indeed has such a theorem for the case of composite $q$.
However, it requires an additional condition on $Y$ apart from $|Y\cdot Y|\leq |Y|^{1+\tau}$.
Roughly speaking, the condition is that if we look at elements of $Y$ modulo a factor
$q_i$ of $q$, they are not too concentrated on any particular element of $\Z_{q_i}$.
See Theorem \ref{thm:bourzpn} for a precise statement.
We show that if $X$ satisfies a certain `list-decodability' condition, $Y$
satisfies the condition required by Bourgain's theorem. For arbitrary $\gamma>0$, we also show that a random source of min-entropy $\gamma n$ and a random affine source of min-entropy $\gamma n$ satisfy the list decodability condition with very high probability.
Thus, our extractor does not work for all additive sources in $\Z_p^n$.
The reduction from the statement about most $0\neq a\in \Z_q$ to all
is similar to the description in the case of $\Zp$.

We show that for the case of affine sources, we do not need a list decodability condition on $X$. A potentially useful tool we develop for this is an XOR lemma that guarantees closeness to uniform under a weaker condition than usual.
The usual setting, described for example by Rao \cite{Rao:bourgain}, is when $N>M$ and for all nontrivial characters $\psi$ on $\Z_N$,
we have $\E_X[\psi(X)] \leq \eps$.  Then there's a simple map $\sigma: \Z_N \to \Z_M$ such that $\sigma(X)$ is close to uniform.
We show that a similar result holds under the weaker assumption that $\E_X[\psi(X)] \leq \eps$ only for characters $\psi$ of the form $\psi(x)=e_n(a\cdot x)$ for $a \in \Z_N^*$, i.e., $(a,N)=1$.

See Section \ref{sec:extrzpn} for full details.

\subsubsection{Extractors for APs and GAPs.}
For the case when the additive source is an AP or GAP in $\Zp^n$,
we give alternate constructions for a wider range of parameters.

For this, we generalize an approach introduced by Gabizon and Raz \cite{GabR}
and used by DeVos and Gabizon \cite{deVG} for constructing extractors for affine sources.
Their approach was to construct a polynomial $f:\Zp^n\to \Zp$
guaranteed to be non-constant on any $k$-dimensional affine subspace.
Given such an $f$ of degree $d$, the Weil bound (Theorem \ref{thm:weiladd})
can be used to construct an extractor for affine sources of dimension $k$ when $p=\Omega(d^2)$.
We show that the same approach works for GAPs:
Suppose we can construct an explicit polynomial $f:\Zp^n\to \Zp$ of degree $d$ that is non-constant
\emph{and of degree larger than one} when restricted to any affine subspace
of dimension $k$. Then we can construct an extractor for GAPs in $\Zp^n$ of dimension
$r\sim k$, assuming $p$ is roughly $\Omega(d^2\cdot \log ^{4} d)$.
See Theorem \ref{thm:poly-to-GAPext} for a precise statement.
Theorem \ref{thm:poly-to-GAPext} follows from a generalization of the Weil bound.
Let us first recall the Weil bound of Theorem \ref{thm:weiladd} says.
Suppose we have a univariate polynomial $f$ over $\Zp$ of degree $d< \sqrt p$.
Suppose $\psi$ is a non trivial additive character of $\Zp$.
Then the character sum,
\[\l|\sum_{t\in \Zp}\psi(f(t))\r|\]
is small; more specifically, it is at most $d\cdot \sqrt q$.
One may ask what happens when the same sum is taken only on the first $s$ elements of $\Zp$.
Perhaps it is significantly larger than $d\cdot \sqrt q$ and becomes smaller only
when running over the whole field?
We show this is not the case.
More precisely, for any $0\leq s <p$
we have
\[\l|\sum_{0 \leq t \leq s-1}\psi(f(t))\r|\leq 16 \log^2 p\cdot \sqrt p\cdot d\]
(see Lemma \ref{lem:GAPexpsum}).
The proof uses a combination of Theorem \ref{thm:weiladd}
and Fourier analysis.
For example, a central step is to bound the Fourier coefficients of
the function $\psi\circ f$ using the Weil bound.

See Section \ref{sec:extrgap} for full details.
\paragraph{Extractors for lines over smaller fields.}
\cite{GabR} used the approach mentioned above to construct extractors for line sources in $\Zp^n$ over fields $\Z_p$ of size $p=\Omega(n^2)$. The main component in their construction was an explicit
polynomial $f:\Zp^n\to \Zp$ of degree $n$ that is non-constant when restricted
to any affine line.
We improve on this and construct a polynomial $f$ of degree $O(\sqrt n)$ that
is non-constant on any affine line.
As a result we get extractors for line sources in $\Zp^n$ when $p=\Omega(n)$.
We sketch the construction of $f$.

\begin{itemize}
\item The first step is to construct a polynomial $g:\Zp^n\to \Zp$ of degree $n$
that is non-constant on any line, and moreover, has the following stronger property:
The restriction of $g$ to any affine line will have degree \emph{exactly} $n$ (rather than just \emph{at most} $n$).
We show that taking $g$ to be a `norm polynomial' insures this property.
\item The second step is to partition the $n$ coordinates into
blocks of ascending size $1,2,\ldots,\ell$ where $\ell = O(\sqrt{n})$.
Now, let $g_i:\Zp^i\to \Zp$ denote the `version' of the polynomial $g$ when
applied to a domain of $i$ coordinates.
We apply $g_i$ to the $i$'th block.
Note that the degree of the $g_i$'s is at most $\mathrm{deg}(g_{\ell}) = \ell = O(\sqrt{n})$.
\item Now we define $f$ to be the sum of the $g_i$'s when applied to the corresponding blocks.
Note that $\mathrm{deg}(f) = O(\sqrt{n})$.
We claim that $f$ is non-constant on any affine line:
Fix any affine line $L$, and fix the maximal $i\in [\ell]$ such
that $L$ is non-constant when restricted to the coordinates of the $i$'th block.
The above-mentioned property of $g$ guarantees that the `$g_i$-summand of $f$' restricted to $L$ will have degree $i$.
All other summands will either be constant or of lower degree.
Thus, $f$ restricted to $L$ is non-constant.
\end{itemize}
See Section \ref{sec:extrlines} for full details.


\subsection{Organization}
In Section~\ref{sec:defn}, we present basic definitions.
In Section~\ref{sec:extrzp}, we
present our deterministic extractor for additive sources in $\Z_p$, and instantiate it in the case of GAPs and Bohr sets.
In Section~\ref{sec:extrzpn}, we give our deterministic extractor for sources in $\Z_p^n$ and again instantiate it in the case of GAPs and Bohr sets.
In Section~\ref{sec:extrlines}, we construct deterministic extractors for lines ($1$ dimensional affine spaces), partial lines in $\Z_p^n$ (APs) and further generalize to GAPs.

\section{Definitions}
\label{sec:defn}
In the following, $p$ will denote a prime number. For $x \in \mathbb{R}$, $\|x\|$ denote the distance to the nearest integer. $e(x)$ denotes the complex number $e^{2 \pi i x}$ and $e_m(x)$ denotes $e^{2 \pi i x/m}$ for any positive integer $m$. To avoid clutter, $e^y$ is written is $\exp (y)$.

\subsection{Probability Distributions and Extractors}
As mentioned earlier,
a set $X$ and a source $X$ shall be used interchangeably where a source $X$ denotes the uniform distribution on the set $X$.

\subsection{Additive Combinatorics}
We now state some standard terminology from additive combinatorics. We refer the reader to \cite{TV} for more details.
In this section, let us fix a finite abelian group $(G,+)$.

\begin{define}[Representation function] Let $A$ be a subset of $G$. For $g \in G$, we define
\[
rep_{A-A}(g)=|A \cap (A+g)|
\]
which is the number of ways to represent $g$ as a difference of two elements in $A$.
\end{define}

\begin{define}[Affine source and line source] A $\de$-affine source in $\Z_p^n$ is an affine source of dimension $\de n$. A dimension $1$ affine source is called a line source.
\end{define}

\begin{define}[Generalized arithmetic progression]\label{dfn:gap} An $(r,s)$-Generalized arithmetic progression (or GAP for short) in $G$ defined is a set of the form
\[
\left\{ b_0+\sum_{i=1}^r a_ib_i:0 \leq a_i \leq s-1 \right\}
\]
for fixed elements $b_0,b_1,\ldots, b_r\in G$ (note that the $a_i$'s are \emph{integers} rather than elements of $G$). We say that the GAP is \textup{proper} if all the $s^r$ sums are distinct.
The dimension of the GAP is $r$.
\end{define}
All GAPs are assumed to be proper in this paper unless mentioned otherwise. In fact, we will see in Section~\ref{sec:extrgap} how to handle general GAPs in $\Fn$.


\begin{define}[$k$-AP and $k$-line] An arithmetic progression of length $k$ (or $k$-AP for short) is a $(1,k)$-GAP. A $k$-AP in $\F_q^n$ is also called a $k$-line.
\end{define}

\begin{define}[$k$-HAP]\label{def:hap} A homogenous arithmetic progression of length $k$ (or $k$-HAP for short) is a $k$-AP with $b_0=0$.
\end{define}

\begin{define}[Bohr set]\label{dfn:bohr} Let $S$ be a set of characters of $G$ and let $\rho >0$. Then we define the Bohr set \[ \Bo(S,\rho)=\{x \in G: \max_{\xi \in S} | \xi(x) |<\rho\}\]
We call the $\rho$ the radius and $|S|$ the rank of the Bohr set. We refer to a Bohr set of rank $d$ and radius $\rho$ as a $(d,\rho)$-Bohr set.
\end{define}
Bohr sets and GAPs are closely related. In particular, any Bohr set contains a large GAP with small dimension \cite[Theorem 7.1]{ruzsa}.

We say that a Bohr set is \textit{regular} if additionally,
\[\Bo(S, \rho (1+\kappa)) \leq (1+100\kappa |S|)\Bo(S, \rho)\]
whenever $\kappa<1/100|S|$. Regular Bohr sets have the property that increasing the radius of the Bohr set by a little does not make the Bohr set very large.
In fact, regular Bohr sets are ubiquitous \cite{TV}, that is every Bohr set is ``close'' to a regular Bohr set. More precisely, for every $S$ and $\eps$, there is $\rho \in [\eps, 2\eps]$ such that $\Bo(S, \rho)$ is regular. In this work, all Bohr sets will be regular Bohr sets.


When $G=\Z_p^n$, we know that the dual of $G$ is isomorphic to $G$. Thus, in this case, we can consider $S \subseteq \Z_p^n$ and the Bohr set
\[ \Bo(S,\rho)=\{x \in \Z_p^n: \max_{\xi \in S} \left\| \frac{\xi \cdot x}{p} \right\| <\rho\}.\]
Here $\| \cdot \|$ denotes the distance to the nearest integer.

Note that if $G$ is a vector space over $\F_q$, then every subspace of $G$ is a Bohr set with radius $1/q$ and rank equal to the codimension of the subspace.
Thus, Bohr sets are generalizations of subspaces and are substitutes for the latter when $G$ has no proper subgroups (e.g., when G= $\Z_p$).
Bohr sets can also be thought of as the inverse image of a cube in $\C^S$ (where $\C$ is the unit circle in $\mathbb{C}$) if one considers the map $x \mapsto (e_p(\xi \cdot x))_{\xi \in S}$. This is justified by the inequality $4\|\theta\| \leq |e(\theta)-1| \leq 2 \pi \|\theta\|$.

\subsection{Characters}

Let $f:\mathbb{Z}_m \rightarrow \mathbb{C}$ be any function. Recall that, for $0 \leq j \leq m-1$, the Fourier coefficients of $f$ are given by \[\widehat{f}\left(j\right)=\frac{1}{m}\sum_{x \in \Z_m} f\left(x\right)\exp\left(-2\pi i jx/m\right).\]
 It is well known that the set of functions $\{\exp\left(2\pi i jx/m\right)\}_{0 \leq j \leq m-1}$ is an orthonormal basis
 for all complex functions defined on $\mathbb{Z}_m$ , and that $f$ can be expressed as \[f\left(x\right)=\sum_{j=0}^{m-1}\widehat{f}\left(j\right)\exp\left(2\pi i jx/m\right).\]
Let us consider $f:\Z_m \rightarrow [0,1]$. Thus, Parseval's identity states that \[\sum_{j=0}^{m-1}\left|\widehat{f}\left(j\right)\right|^2 =\frac{1}{m}\sum_{x \in \Z_m} f(x)^2  \leq 1.\]

\noindent \textbf{Exponential/Character sums to extractors.}
Throughout the paper, $\psi$ and $\chi$ denote additive and multiplicative characters respectively and $\psi_0$ and $\chi_0$ denote the trivial additive and multiplicative characters respectively. We let $e_n(x)$ denote $e^{2\pi i x/n}$.
We now state two lemmas that gives a black box construction of deterministic extractors from exponential/character sums. Note that we use the term exponential sum for additive characters and character sums for multiplicative characters.

The following lemma is for exponential sums.
\begin{lem}\label{lem:addlemZpn}Let $X \subseteq \F_{p}^n$. If $ \left| \frac{1}{|X|}\sum_{x \in X}\psi(x) \right|<\eps$ $\forall \psi \neq \psi_0$, then there exists an efficient $\sigma:\F_{p}^n \rightarrow \F_{p}^m$ such that \[|\sigma(X)-U|<\eps \sqrt{p^m}\]
\end{lem}

We state a similar lemma that works for cyclic groups. A proof of this can be found in \cite{Rao}.
\begin{lem}\label{lem:addlem}Let $X \subseteq \Z_N$. If $\left| \frac{1}{|X|}\sum_{x \in X}\psi(x) \right|<\eps$ $\forall \psi \neq \psi_0$, then there exists an efficient $\sigma:\Z_N \rightarrow \Z_M$ such that \[|\sigma(X)-U|<\eps \sqrt{M}\log N +O(M/N)\]
\end{lem}

The next lemma is for character sums.
\begin{lem}\label{lem:multlem}Let $X \subseteq \F_{p^n}^*$. If $\left| \frac{1}{|X|}\sum_{x \in X}\chi(x) \right|<\eps$ $\forall \chi \neq \chi_0$, then there exists an efficient $\sigma:\F_{p^n}^* \rightarrow \F_{p^m}^*$ such that \[|\sigma(X)-U|<\eps \sqrt{p^m}\]
\end{lem}
\begin{proof}
Without loss of generality let us assume that $m$ divides $n$. If not, we can always append $0$'s to increase the dimension by a factor of at most $2$.
We have the following standard claim.
\begin{claim}Let $X$ be a distribution on $G$ such that $|\E[\chi(X)]| \leq \eps$ $\forall \chi \neq \chi_0$. Then, $X$ is $\eps \sqrt{|G|}$ close to $U$.
\end{claim}

With the above claim, we define $\sigma:\F_{p^n}^*\rightarrow \F_{p^m}^*$ as $\sigma(x)=x^{\frac{p^n-1}{p^m-1}}$. Now,
\begin{claim}Given a nontrivial multiplicative character $\Psi$ of $\F_{p^m}^*$, $\Psi \circ \sigma$ is a nontrivial multiplicative character of $\F_{p^n}^*$.
\end{claim}

Thus, by hypothesis, $\left| \frac{1}{|X|}\sum_{x \in X}\Psi \circ \sigma(x) \right| <\eps$. Therefore, $ \left| \E_{\sigma(X)}\Psi(\sigma(X)) \right|<\eps$.
Thus, $\sigma(X)$ is $\eps p^{m/2}$ close to $U$.
\end{proof}

\noindent \textbf{The Riemann Hypothesis for curves over finite fields.}
In 1948 Weil \cite{Wei1948} proved the celebrated \textit{Riemann Hypothesis for curves over finite fields}.
A consequence of Weil's result is a bound for exponential and character sums over low degree polynomials over a finite field. We state it below.
The theorems can also be found in \cite{Sch1976}.
\begin{thm}[Weil's bound]\label{thm:weiladd}
Let $\psi$ be a nontrivial additive character of $\F_q$. Let $f(t) \in \F_q[t]$ be a polynomial of degree $m$. Let $gcd(m,q)=1$. Then
\[ \left|\sum_{t \in \F_q}\psi(f(t)) \right|\leq mq^{1/2}.
\]
\end{thm}

\begin{thm}[Weil's bound]\label{thm:weilmult}
Let $\chi$ be a nontrivial additive character of $\F_q$. Let $f(t) \in \F_q[t]$ be a polynomial of degree $m$.
Suppose that $f(t)$ is not of the form $cg(t)^m$ for any $c \in \Fq$ and $g(t) \in \Fq[t]$. Then
\[ \left|\sum_{t \in \F_q}\chi(f(t)) \right|\leq mq^{1/2}.
\]
\end{thm}

\section{Extractors for additive sources in $\Z_p$}\label{sec:extrzp}

We now state our extractors for additive sources in $\Z_p$.
\subsection{An extractor for additive sources}
Our main theorem for $\Z_p$ (Theorem~\ref{THM:zp}) follows from the following theorem.
\begin{thm}\label{thm:extrstruc} Fix any $\de>0$ and positive constant $C$. There exists $p_0 \in \mathbb{N}$ such that for all primes $p \geq p_0$ the following holds.  There is an explicit efficient $\eps$-extractor $\ext:\Z_p \rightarrow \{0,1\}^m$, for ($\alpha, \beta, \tau$)-additive sources of entropy rate $\de$ in $\Z_p$ where $\de \beta\geq 2t\log_p (1/\alpha) + \de/C$, $\eps=3\alpha 2^{m/2} \log p +O(2^m/p)$ and $\tau,t$ are constants depending only on $\de$ and $C$.
\end{thm}
\begin{proof}[Proof of Theorem~\ref{THM:zp}] Let $C=2/\beta$ and $\gamma=\frac{\beta\de}{4t}$. Then the hypothesis of the above theorem is satisfied if $\alpha>p^{-\gamma}$. The $2^m/p$ term can now be dropped by assuming without loss of generality $\gamma<1/2$. Now, since, any $(\alpha,\beta,\tau)$-additive source is an $(\alpha',\beta,\tau)$-additive source for $\alpha<\alpha'$, this finishes the proof.
\end{proof}

\noindent The above theorem follows from Lemma \ref{lem:zpaddfriendly} and Lemma ~\ref{lem:addlem}.

\noindent Before we state Lemma \ref{lem:zpaddfriendly} and prove it, we state the following theorem.

\begin{thm}[{\cite[Theorem 1']{Bou:diff}}]\label{thm:bour}For all $Q \in \Z_+$, there is $\tau>0$ and $t \in \Z_+$ such that if $H\subseteq \F_p^*$ satisfies $|H \cdot H|<|H|^{1+\tau}$, then
\[\frac{1}{p}\sum_{a=0}^{p-1}\left|\sum_{x \in H}e_p(a x)\right|^{2t}<|H|^{2t}\left(C_Q|H|^{-Q}+p^{-1+1/Q}\right)\]
\end{thm}

\begin{lem}\label{lem:zpaddfriendly}There exists $p_0 \in \mathbb{N}$ such that for all primes $p \geq p_0$ the following holds. There exists an efficient $f:\Z_p \rightarrow \Z_q$ (for $q=o(p^6)$) such that if $\de>0$ is arbitrary and $C$ is an arbitrary large constant, then there exist $\tau(\de,C)>0$ and $t(\de,C)$ such that if
\begin{itemize}
\item  $X$ is an ($\alpha, \beta, \tau$)-additive set of entropy rate $\de$ in $(\Z_p,+)$ ,
\item $\beta \de\geq 2t\log_p (1/\alpha) + \de/C$,
\end{itemize}
then, for all $\xi \in \Z_q \setminus \{0\}$,
\[\left|\sum_{x\in X}e_q(\xi f(x))\right|<3\alpha |X|\].
\end{lem}
\begin{proof}Let $q$ be the smallest prime such that $q = 1 \pmod p$. By Linnik's theorem, such a $q$ exists and $q=O(p^{5.2})$. Let $g$ be an element of $\Z_q^*\subset \Z_q$ such that $ord(g)=p$. Now define $f:\Z_p \rightarrow \Z_q$ as follows. Let $f(x)=g^x$.
Let $Y=f(X)$. Let $S=Sym_{1-\alpha}(X)$.
\begin{claim}\label{clm:zpadditive}$Y$ is a $(\alpha, \beta,\tau)$-multiplicative set of entropy rate $\de/6$ in $(\Z_q^*, \times)$.\end{claim}
\begin{proof}
We first note that $f:X \rightarrow Y$ is injective. To see that, suppose for $x,y \in \Z_p$ we have $g^x=g^y$.
This implies $g^{x-y}=1$. Since $\textup{ord}(g)=p$, we have $x \equiv y \pmod{p}$.
Now, since $f$ is injective, $|Y|\geq p^{\de} \geq q^{\de/6}$. Also, we claim that for each $a \in S$, $rep_{Y \cdot Y^{-1}}(f(a))\geq (1-\alpha)|Y|$: This is because if $a\equiv x-x'\pmod{p}$ for $x,x' \in X$, then $g^a = g^{x-x'}=f(x)/f(x') \in Y\cdot Y^{-1}$, where the first equality again uses the fact that $\textup{ord}(g)=p$.   So $\left|  Y \cap (Y \cdot f(a)) \right|=\left|  X \cap (X+a) \right| \geq (1-\alpha) |Y|$. Now, observe that $|f(S)|\geq |X|^{\beta} = |Y|^{\beta}$. Finally, we show that $|Y \cdot Y| \leq |Y|^{1+\tau}$. This follows from
the fact that for $a,b \in X$, $f(a)\cdot f(b) = f(a+b)$ and therefore $|Y\cdot Y| = |X+X|$.
\end{proof}

We now continue with the proof. Let $Q=6C/\de$.
Let $M=\max_{\xi \neq 0}\left|\sum_{y\in Y}e_q(\xi y)\right|$ and let $\xi$ attain $M$.
\begin{claim}$M<3|Y|\alpha$.
\end{claim}
\begin{proof}Suppose $M\geq 3|Y|\alpha$.  Consider any $\xi' \in f(S)$. Then,
\begin{eqnarray*}
\l|\sum_{y \in Y}e_q(\xi' \xi y)\r| &=&\l|\sum_{y \in \xi' Y}e_q(\xi y)\r|
\geq \l|\sum_{y \in Y}e_q(\xi y)\r|-2(|Y|-|Y \cap \xi' Y|)\\
&\geq&M-2|Y|\alpha \geq |Y|\alpha.
\end{eqnarray*}
Since this lower bound holds for any $\xi' \in f(S)$, we have \[|f(S)|\l|Y\r|^{2t}\alpha^{2t} \leq \sum_{\xi}\l|\sum_{y \in Y}e_q(\xi y)\r|^{2t}.\]
Therefore, since $|f(S)| \geq p^{\de \beta}$, and $Y$ satisfies the hypothesis of Theorem \ref{thm:bour}, we have \[p^{\de \beta}\alpha^{2t} \leq q\left(C_Q\l|Y\r|^{-Q}+q^{-1+1/Q}\right)  <  p^{\de/C}\] for large enough $p$.
But this implies $\de \beta-2t\log_p (1/\alpha)<\de/C$ which is a contradiction.
Thus, we have $\max_{\xi \neq 0}\l|\sum_{y\in Y}e_q(\xi y)\r|<3|Y|\alpha$.
\end{proof}
This implies \[\l|\sum_{x\in X}e_q(\xi f(x))\r|<3|X|\alpha.\]
\end{proof}

We proceed to formally show that GAPs and Bohr sets are indeed additive sources in $\Z_p$.
We then use Theorem \ref{thm:extrstruc} to derive corollaries on these types of sources.

\subsection{Application to GAPs and Bohr sets}
We first show that a GAP source is an additive source with the appropriate parameters.
\begin{lem}\label{lem:gapzp}For all $\eps >0$, there exists $c,n_0 \in \N$ such that for all prime $p \geq n_0$ the following holds. If $\de \geq c/\log p$, then an $\l(r,p^{\de}\r)$-GAP source is a $\l(r/p^{0.9 \de}, 0.1, \eps \r)$-additive source of entropy rate $ \de r$ in $(\Z_p,+)$.\end{lem}
\begin{proof}Let $X$ be the $\l(p^{\de},r\r)$-GAP source defined by $X=\{b_0+\sum_{i=1}^r a_ib_i:0 \leq a_i \leq s-1\}$ where $s=p^{\de}$ and let it be $\l( \alpha_0, \beta_0, \tau_0\r)$-additive.
It is easy to see that the entropy rate is $\de r$. The lemma now follows from a series of claims.
\begin{claim}$\tau_0 \leq \eps$ for all $\eps>0$.\end{claim}
\begin{proof}Note that \[X+X=\{2b_0+\sum_{i=1}^r a_ib_i:0 \leq a_i \leq 2s-2\}\] Therefore, $|X+X|\leq 2^rs^r=2^r|X|$ since $X$ is a proper GAP. Now, $2^r<|X|^{\tau_0}$ iff $s^{\tau_0}>2$ which is true for constant $\tau_0=\eps$ since $s = p^{\de} \geq 2^c$.
\end{proof}

\begin{claim}$\alpha=r/p^{0.9 \de}$ and $\beta=0.1$.
\end{claim}
\begin{proof}
Consider the set $S=\{b_0+\sum_{i=1}^ra_ib_i:0 \leq a_i < s^{0.1}\}$. Now fix an arbitrary $x \in S$. Then, \[X \cap (X+x) \supseteq \{b_0+\sum_{i=1}^r  a_ib_i:s^{0.1} \leq a_i < s\}\]
Therefore, $\l|X \cap (X+x)\r| \geq \l(s-s^{0.1}\r)^r=|X|\l(1-1/s^{0.9}\r)^r>|X|\l(1-r/s^{0.9}\r)$. Also, we have $|S| \geq |X|^{0.1}$. This proves the claim.
\end{proof}
\end{proof}

Note that the requirement of $\de \geq c/\log p$ merely means that the sides of the GAP are $p^{\de}=\Omega(1)$ in length.

Next, we show that a Bohr set is an additive source with the appropriate parameters.
We will use the following lemma from \cite{TV} for the group $G$. As before let $S$ be a set of frequencies of $G$.
\begin{lem}[Lemma 4.20 \cite{TV}]\label{lem:bohrtv}$|\Bo(S,\rho)| \geq \rho^{|S|}|G|$ and $|\Bo(S,2 \rho)|\leq 4^{|S|}|\Bo(S,\rho)|$.
\end{lem}

We are now ready to prove our lemma about Bohr sets.
\begin{lem}\label{lem:bohrG}Let $\beta, \eps, \de, \rho > 0$ be arbitrary and $S \subseteq \widehat{G}$ be a set of frequencies. Let $B=\Bo(S,\rho)$ in $G$ where
$d=|S|$. Let $0 \leq \kappa \leq \frac{1}{100d}$. A Bohr source is a $\l(100\kappa d, \beta, \eps\r)$-additive source of entropy rate $\de$ in $G$ whenever $|G| \geq \max \l\{\l(\frac{4^{1/\eps}}{\rho}\r)^d, \l(\frac{1}{\rho}\r)^{d/1-\de}, \l(\frac{1}{\kappa \rho}\r)^{d/1-\beta}\r\}$.\end{lem}
\begin{proof}
$|B| \geq \rho^d |G|$ by Lemma \ref{lem:bohrtv} and by the hypothesis, we have $|B| \geq |G|^{\de}$.

To see that $B$ has small doubling, observe that $B+B \subseteq \Bo(S, 2 \rho)$ and therefore, using the fact $|\Bo(S,2\rho)|\leq 4^d \Bo(S,\rho)$ (Lemma \ref{lem:bohrtv}) we have $|B+B| \leq 4^d |B|<|B|^{1+\eps}$. The last inequality is true because $|B| \geq \rho^d |G|$ (Lemma \ref{lem:bohrtv}) and $\rho^d |G| >4^{d/\eps}$ by hypothesis.

We now argue the presence of large symsets in $B$. Let $Y=\Bo(S, \kappa \rho)$. Fix $y \in Y$. For any $x \in \Bo(S, (1-\kappa)\rho)$, $x+y \in B$. Therefore, \[|B \cap (y+B) |\geq |\Bo(S, (1-\kappa)\rho)|\geq (1-100\kappa d)|B|\]
This is because we consider regular Bohr sets. Also, $|Y|\geq (\kappa \rho)^d |G| > |G|^{\beta} \geq |B|^{\beta}$ by the hypothesis and Lemma \ref{lem:bohrtv}. This finishes the proof.
\end{proof}

\paragraph{GAP sources.}
We first restate our corollary for GAP sources.\\ \\
\textbf{Corollary~\ref{cor:zpgap}.} For all $\de_0>0$, there exists $c,p_0 \in \mathbb{N}, \de_0>0$ such that for all primes $p \geq p_0$ the following holds. There exists an explicit efficient $\eps$-extractor $\ext:\Z_p \rightarrow \{0,1\}^m$, for ($r,p^{\de}$)-GAP sources (of entropy rate $\de_0=\de r$) in $\Z_p$ where $p^{\de} \geq c$, $r \geq C_{\de_0}$ (where $C_{\de_0}$ is a constant depending on $\de_0$ only) and $\eps=3\l(r/p^{0.9\de}+p^{-1/2}\r) 2^{m/2} \log p$.

\begin{proof}By Lemma \ref{lem:gapzp}, an ($r,p^{\de}$)-GAP source is a $\l(r/p^{0.9 \de}, 0.1, \tau\r)$-additive source of entropy rate $\de r$ in $\Z_p$ for all $\tau>0$. We will apply Theorem \ref{thm:extrstruc} with $C=20$ and $\de'=\de r$ which is a constant. Therefore, $t$ and $\tau$ from the theorem conclusion are also constants. Note that we already have $s>2^c$ by Lemma \ref{lem:gapzp}. Now, the second condition in Theorem \ref{thm:extrstruc} is equivalent to $0.1\de r>2t \log_p \l(p^{0.9\de}/r \r)+\de r/20$ is satisfied if $r \geq 36t$. We also drop the $2^m/p$ by putting a $p^{-1/2}$ term similar to Theorem~\ref{THM:zp}. This finishes the proof.
\end{proof}

Note that the above extractor works for GAPs with sides as small as superconstant. It works as long as the total volume of the GAP exceeds $p^{\Omega(1)}$ which clearly improves upon a standard convex combination argument by extracting of each individual AP as that would need at least $p^{\Omega(1)}$ entropy along each side.

\paragraph{Bohr sources.}
We now restate our corollary for Bohr sources. \\ \\
\textbf{Corollary~\ref{cor:zpbohr}.}  Let $\rho, \alpha>0$ and $S \subseteq \Z_p$ with $|S|=d$ be arbitrary. Then for prime $p=\Omega\l(\l(\frac{d}{\alpha}\r)^d\r)$, there exists an explicit efficient $\eps$-extractor $\ext:\Z_p \rightarrow \{0,1\}^m$, for ($d, \rho$)-Bohr sources of entropy rate $\de$ in $\Z_p$ where $\eps=\l(3\alpha+p^{-\Omega(1)}\r) 2^{m/2} \log p  $.

\begin{proof}
By Lemma \ref{lem:bohrG}, any $\Bo(S,\rho)$ is a $\l(100\kappa d, \beta, \eps  \r)$-additive source of entropy rate $\de$ in $\Z_p$ whenever $p \geq \max \left\{ \l(\frac{4^{1/\eps}}{\rho}\r)^d, \l(\frac{1}{\rho}\r)^{d/1-\de}, \l(\frac{1}{\kappa \rho}\r)^{d/1-\beta} \right\}$. The statement follows from the lower bound on $p$ and Theorem~\ref{THM:zp}. 
\end{proof}

Note that the upper bound on $\de$ is reasonable because as $\de$ increases the Bohr structure keeps fading away. Hence we cannot extract from arbitrarily large Bohr sets.

We can extract more randomness from GAPs under the Paley Graph Conjecture which we show in Appendix \ref{app:pgc}.

\section{Extractors for additive sources in $\Z_p^n$}\label{sec:extrzpn}

We now state our extractors for additive sources in $\Z_p^n$.

\subsection{GAPs and Bohr sets}
We first show that a GAP source is an additive source with the appropriate parameters.
\begin{lem}\label{lem:gapzpn}For all $\eps >0$, there exists $c,n_0 \in \N$ such that for all prime $p \geq n_0$ the following holds. If $\de \geq \l(C/\log p\r)$, then an $\l(r=\mu n,p^{\de}\r)$-GAP source is a $\l(\mu n/p^{0.9 \de}, 0.1, \eps  \r)$-additive source of entropy rate $\de \mu$ in $(\Z_p,+)$.\end{lem}
\begin{proof}Let $X$ be the $\l(r=\mu n, p^{\de}\r)$-GAP source defined by $X=\{b_0+\sum_{i=1}^r a_ib_i:0 \leq a_i \leq s-1\}$ where $s=p^{\de}$ and let it be $\l( \alpha_0, \beta_0, \tau_0  \r)$-additive. 
It is easy to see that the entropy rate is $\mu \de$. The lemma now follows from a series of claims.
\begin{claim}$\tau_0 \leq \eps$ for all $\eps>0$.\end{claim}
\begin{proof}Note that \[X+X=\{2b_0+\sum_{i=1}^r a_ib_i:0 \leq a_i \leq 2s-2\}\] Therefore, $|X+X|\leq 2^rs^r=2^r|X|$ since $X$ is a proper GAP. Now, $2^r<|X|^{\tau_0}$ iff $s^{\tau_0}>2$ which is true for constant $\tau_0=\eps$ since $s = p^{\de} \geq 2^c$.
\end{proof}

\begin{claim}$\alpha=r/p^{0.9 \de}$ and $\beta=0.1$.
\end{claim}
\begin{proof}
Consider the set $S=\{b_0+\sum_{i=1}^ra_ib_i:0 \leq a_i < s^{0.1}\}$. Now fix an arbitrary $x \in S$. Then, \[X \cap (X+x) \supseteq \{b_0+\sum_{i=1}^r  a_ib_i:s^{0.1} \leq a_i < s\}\]
Therefore, $\l|X \cap (X+x)\r| \geq \l(s-s^{0.1}\r)^r=|X|\l(1-1/s^{0.9}\r)^r>|X|\l(1-r/s^{0.9}\r)$. Also, we have $|S| \geq |X|^{0.1}$. This proves the claim.
\end{proof}
\end{proof}

Note that the requirement of $\de \geq \l(C/\log p\r)$ merely means that the sides of the GAP are $p^{\de}=\Omega(1)$ in length.

We have already shown in Lemma \ref{lem:bohrG} that Bohr sets are additive sources. We now proceed with the main theorem of this section.

\subsection{Extractor for additive sources}

We say that a set $X$ is $(r,B)$-list decodable if for any arbitrary $r$ indices $i_1, \cdots i_r$, $c_j \in \Z_p$ for $j \in [r]$, $\l|X_{x_{i_1}=c_1, \ldots x_{i_r}=c_r}\r| \leq B$.
We now state the main theorem of this section.

\begin{thm}\label{thm:extrstruczpn} There exists $p_0, L_0 \in \mathbb{N}$ such that for all $L \geq L_0$ and primes $p \geq p_0$ the following holds. Let  $\gamma, \kappa>0$ be arbitrary. There exists an efficient $\eps$-extractor $\ext:\Z_p^n \rightarrow \{0,1\}^m$ for $(\alpha, \kappa L/\de, \tau)$-additive sources of entropy rate $\de$ in $(\Z_p^n,+)$ where $n \leq p^{L-2}$, $X$ is $(r, p^{-r\cdot\gamma\cdot L}\cdot |X|)$-list decodable for every $\tau n/L \leq r\leq n$ and $\eps<  \l(3\alpha+|X|^{-\tau}\r) 2^{m/2} \log p^n +O(2^m/p^n)$ where $\tau$ is a constant depending on $\gamma$ and $\kappa$.
\end{thm}

The following remarks show that the list decodability condition is not too restrictive.
\begin{remark}[Min-entropy $(1-\eps')n$ is list decodable]\label{rem:zpn1}
In fact, for high enough min-entropy, we can now eliminate the list decodability
assumption altogether.
Given $L$, choose $\gamma=1/(4L)$, (and fix some $\kappa>0$ as in Theorem~\ref{thm:extrstruczpn}).
Now this fixing of $\gamma$ and $\kappa$ also fixes some $\tau>0$.
Denote $a = \tau n/L$.
We claim the thm can now be applied to {\emph any} source of entropy $k=n- a/2$.
Fix such a source $X$. For $r>a$, and any fixing of any $r$ coordinates, the
corresponding list will be of size at most
$p^{n-r}$. We need to show that this is smaller than $p^{-L\gamma r }|X|
= p ^{-r/4 + n- a/2}$
It can be checked that this is indeed the case
$$n-r< -r/4 +  n- a/2  \qquad iff  \qquad  (3/4)r > a/2$$
\end{remark}

\begin{remark}[Random set of min-entropy $\eps' n$ is list decodable] \label{rem:zpn2} A random set $X$ of size $|X|>p^{2 L \gamma n}$ satisfies the list decodability condition with high probability. To see this,
fix $r>\tau n/L$, a set of indices $S$ of size $r$, field values $c_1,\ldots ,c_r$ for those indices and a subset $W$ of the set of size $B=p^{-\gamma L r}|X|+1$. The probability that all of $W$ has the property that the coordinates in $S$ get values $c_i$'s is $(1/p^r)^B$. A union bound over all $W$ gives
$\binom{|X|}{B}(1/p^r)^B<(|X|e/B)^B (1/p^r)^B<<1/p^{0.9rB}$ as $\gamma$ is arbitrarily small. An outer round of union bound over each of the $p^r$ settings of $c_i$'s and $S$ is too mild to boost up the error probability for large $p$.
\end{remark}

\begin{remark}[Random affine source of min-entropy $\eps' n$ is list decodable] \label{rem:zpn3} Let $\gamma<1/L$ be an arbitrary small constant. A subspace $X$ of dimension $k>2\gamma L n$ defined by a random $k \times n$ matrix satisfies the list decodability condition. Indeed, let $G$ be the random $k \times n$ matrix. We know that for any submatrix $C$ of $r$ columns in $G$,  $C$ has rank at least $\gamma r L$ with high probability. To see this, fix a subset of $r$ column indices. Let $a=\gamma r L$. Note that $a < k/2$. Let $C$ be the submatrix of $G$ defined by the $r$ columns. Then, $\Pr[\rank(C)<a]< \binom{r}{a} p^{a-k}$. Here we are saying that some choice of the a columns of $C$ will be linearly independent and then using the bound that a random $s \times t$ matrix has full rank with probability roughly at least $1-p^{s-t}$  (for $s<t/2$). Continuing with the analysis, $\binom{r}{a}p^{a-k}<2^rp^{-k/2}<2^np^{-k/2}$. Taking a union bound over the choice of $r$ columns, we incur another factor of $2^n$, and taking $p\geq 5^{2n/k}$ gives error at most $(4/5)^n$, finishing the proof.
Let us continue with the proof. Fix $r$ coordinates $i_1,\ldots ,i_r$. Let $c_1,\ldots ,c_r$ be $r$ values in the field.  Since the corresponding submatrix $C$ has rank at least $\gamma r L$, the number of strings in $X$ which are $c_j$ in coordinate $i_j$, $j=1, \ldots ,r$, is at most $p^{-\gamma L r}|X|$. This satisfies the list decodability condition for all $r$.  \end{remark}

The theorem follows from Lemma \ref{lem:charsumzn} and the following lemma.
\begin{lem}\label{lem:zpnaddfriendly} There exists $p_0, L_0 \in \mathbb{N}$ such that for all $L \geq L_0$ and primes $p \geq p_0$ the following holds. There exists an efficient $f:\Z_p^n \rightarrow \Z_q$ (for $p^n < q < p^{Ln}$) such that if $\gamma, \kappa>0$ are arbitrary, then there exists $\tau>0$ such that if
\begin{itemize}
\item $X$ is an ($\alpha, \frac{\kappa L}{\de}, \tau  $)-additive source of entropy rate $\de$ in $(\Z_p^n,+)$
\item $n \leq p^{L-2}$
 \item for every integer $\tau n/L \leq r\leq n$, $X$ is $(r,p^{-r\cdot\gamma\cdot L}\cdot |X|)$-list decodable,
\end{itemize}
Then, for all $\xi \in \Z_q \setminus \{0\}$,
\[\l|\sum_{x\in X}e_q(\xi f(x))\r|<3\max \{\alpha, 1/|X|^{\tau}\}  |X|  \]
\end{lem}
\begin{proof}Let $q_1, q_2, \cdots q_n$ be $n$ distinct primes such that for all $i$, $q_i \equiv 1 \pmod{p}$. This is guaranteed by the following consequence of (the proof of) Linnik's theorem. 
\begin{claim}There exist constants $L_0>0$ and $n_0 \in \N$ such that  for all $p \geq n_0, L \geq L_0$, the size of the set 
\[
\{q:q \equiv 1 \pmod{p}, q \leq p^L \}
\] is at least $p^{L-2}$.
\end{claim}
\begin{proof} Define \[\theta(x,p)= \sum_{\substack{k \, \textup{prime}, \, k \leq x,\\ k \equiv 1 \pmod{p}}} \log k.\]
By \cite[Corollary 18.8]{ik}, there is a constant $L_0$ (known as \textit{Linnik's constant}) such that for all $p$ sufficiently large and $x \geq p^{L_0}$, 
we have 
\[
\theta(x,p) \geq \frac{Cx}{p^{1/2}\phi(p)} \geq \frac{Cx}{p^{3/2}}
\]
for some constant $C$, where $\phi(n)$ is Euler's totient $\phi$ function. Let 
\[
\pi(x,p)= \sum_{\substack{k \, \textup{prime}, \, k \leq x,\\ k \equiv 1 \pmod{p}}} 1.
\] 
Then $\pi(x,p) \log x \geq \theta(x,p) \geq  \frac{Cx}{p^{3/2}}$. Thus,  $\pi(x,p) \geq  \frac{Cx}{p^{3/2}\log x}$. 
If $x=p^L$ for $L \geq L_0$, then this is clearly $\geq p^{L-2}$.
\end{proof}

Thus, by the above, we have for all $i$, $q_i<p^{L}$. Also, let $g_i$ generate the order $p$ subgroup in $\Z_{q_i}^*$. Define two maps $\phi_1, \phi_2$ as follows. 
Let $\phi_1:\Z_p^n \rightarrow \prod_{i \in [n]}\Z_{q_i}$ be defined by \[\phi_1(x_1,x_2,\cdots x_{n})= (g_1^{x_1},\cdots g_{n}^{x_{n}})\]
and for $q=\prod_{i \in [n]}q_i$, let $\phi_2:\prod_{i \in [n]}\Z_{q_i} \rightarrow \Z_q$ be defined by \[\phi_2(y_1,\ldots y_n)=\sum_{i=1}^n  y_i \frac{q}{q_i}\l[\l(\frac{q}{q_i}\r)^{-1}\r]_{q_i} \in \Z_q\]
where $[x^{-1}]_p$ is the inverse of $x$ in $\Z_p^*$. Note that $\phi_2$ is the Chinese remaindering map.

Define function $f$ as follows. \[f:\Z_p^n \rightarrow \Z_q\]\[x \mapsto \phi_2 \circ \phi_1 (x)\] Let $Y=f(X)$. $|Y| \geq q^{\de/L}$.
In the following, define for $q', q$, $\pi_{q'}:\Z_q \rightarrow \Z_{q'}$ by $\pi_{q'}(x)=x \pmod {q'}$. Then, for $\xi \in \Z_{q'}$, we define $\pi^{-1}_{q'}(\xi)=\{x \in \Z_q:\pi_{q'}(x)=\xi\}$. We now have the following claim.
\begin{claim}If $q'|q$, and $q'>q^{\tau}, \xi \in \Z_{q'}$, then $|Y \cap \pi_{q'}^{-1} (\xi) | < (q')^{-\gamma} |Y|$.
\end{claim}
\begin{proof}
Without loss of generality, let $q'=\prod_{i=1}^r  q_i$ for $r \leq n$. As $q'>q^{\tau}$, we have $p^{Lr}>q^{\tau}>p^{n \tau}$, since each $q_i \geq p$. 
Therefore, $r>\frac{n\tau}{L}$.
Next, we need the following claim.

\begin{claim}Let $q'$ be as above. Given any $\xi \in \Z_{q'}$, let $\xi_i=\xi \pmod {q_i}$ for $1 \leq i \leq r$. Then for $Y \subseteq \Z_q$, we have \[\l|Y \cap \pi_{q'}^{-1}(\xi)\r|=\l|X_{x_1= \log_{g_1}\xi_1, \ldots x_r=\log_{g_r}\xi_r}\r|\]
\end{claim}

\begin{proof}
Note that $Y \cap \pi_{q'}^{-1}(\xi)=\{y \in Y: y \pmod {q'}=\xi\}$. The condition $y \pmod {q'}=\xi$ can be re-written as $y_i=\xi_i \pmod {q_i}$ as $q_i | q'|q$. Therefore, we have 
\begin{eqnarray*}
\{y \in Y: y \pmod {q'}=\xi\} 
&=&\{y \in Y: y_i \pmod {q_i}=\xi_i , \ 1 \leq i \leq r\}\\
&=&\{x \in X: g_i^{x_i} \pmod {q_i}=\xi_i ,\  1 \leq i \leq r\}\\
&=&\{x \in X: x_i = \log_{g_i}\xi_i ,\  1 \leq i \leq r\}
\end{eqnarray*}
\end{proof}

Now, by the hypothesis, for $r \geq n\tau/L$, we have for any $c_1, \ldots, c_r$, 
\begin{eqnarray*}
\l|X_{x_1=c_1, \ldots x_r=c_r}\r| & \leq &p^{-r\cdot\gamma\cdot L}\cdot |X|\\
&<& (q')^{-\gamma}\cdot |Y| 
\end{eqnarray*}
as $p^{rL}>q'$. This finishes the proof.
\end{proof}

Next, we have the following.
\begin{claim}$|Y.Y|<|Y|^{1+\tau}$\end{claim}
\begin{proof} This follows because $f$ is an one-one function from $\l(\Z_p^n,+\r)$ into $(\Z_q^*,*)$.
\end{proof}
Next, we have the following claim.
\begin{claim}$\l|Sym_{1-\alpha}(Y)\r| \geq q^{\kappa}$.\end{claim}
\begin{proof}The proof follows because $f$ is an one-one function from $\l(\Z_p^n,+\r)$ into $(\Z_q^*,*)$. It is similar to the proof of Claim~\ref{clm:zpadditive}. This would show that $\l|Sym_{1-\alpha}(Y)\r| \geq  |Y|^{\beta=\kappa L/\de}=p^{n\kappa L}>q^{\kappa}$.
\end{proof}

We now need the following theorem due to Bourgain bounding the number of large Fourier coefficients.
\begin{thm} [{\cite[Corollary 3]{Bour2}}]\label{thm:bourzpn}
Given $\gamma, \kappa >0$, there is $\tau>0$ such that the following holds. Let $q$ be an arbitrary modulus and $H \subseteq \Z_q^*$ satisfy
\begin{itemize}
\item If $q'|q$, and $q'>q^{\tau}$, $\xi \in \Z_{q'}$, then $|H \cap \pi^{-1}_{q'}(\xi)|<(q')^{-\gamma}|H|$
\item $|H.H|<|H|^{1+\tau}$.
\end{itemize}
Then $\l|\{\xi \in \Z_q:\l|\sum_{x \in H}e_q(\xi x)\r|>|H|^{1-\tau}\}\r|<q^{\kappa}$.
\end{thm}

Let $M=\max_{\xi \neq 0}\l|\sum_{y\in Y}e_q(\xi y)\r|$ and let $\xi$ attain $M$. Let $S=\{x \in \Z_p^n: |rep_{X-X}(x)|>(1-\alpha)|X|\}$.
\begin{claim}$M<3|Y|\alpha$
\end{claim}
\begin{proof}Suppose $M\geq 3|Y|\alpha$.  Consider any $\xi' \in f(S)$. Note that $|Y \cap \xi'Y| \geq (1-\alpha)|Y|$. Then,
\begin{eqnarray*}
\l|\sum_{y \in Y}e_q(\xi' \xi y)\r| &=& \l|\sum_{y \in \xi' Y}e_q(\xi y)\r|\\
&\geq&\l|\sum_{y \in Y}e_q(\xi y)\r|-2(|Y|-|Y \cap \xi' Y|)\\
&\geq& M-2|Y|\alpha\\
&\geq& |Y|\alpha\\
&\geq&|Y|^{1-\tau}.
\end{eqnarray*}

Since the above lower bound holds for any $\xi' \in f(S)$, we have a contradiction to Theorem \ref{thm:bourzpn} above as $|f(S)|=|X|^{\beta}=p^{\de n(\kappa L/\de)}>q^{\kappa}$.
Thus, we have $\max_{\xi \neq 0}\l|\sum_{y\in Y}e_q(\xi y)\r|<3|Y|\alpha$
\end{proof}

This implies $\l|\sum_{x\in X}e_q(\xi f(x))\r|<3\alpha |X|$, as desired.
\end{proof}

\subsection{Application to GAPs and Bohr sets}

We first state our corollary for GAP sources.

\begin{cor}Let $C>0$ be arbitrary. There exists $p_0, L_0 \in \N$ such that for all $L \geq L_0$ and primes $p \geq p_0$ the following holds. Let $\de , \mu>0$ be arbitrary. There exists an efficient $\eps$-extractor $\ext:\Z_p^n \rightarrow \{0,1\}^m$ for $(\mu n, p^{\de})-GAP$ sources (of entropy rate $\mu \de$) in $\Z_p^n$ where $n \leq p^{L-2}$, $X$ is $(\tau n/L, |X|^{1-1/C})$-list decodable and $\eps< \l(3\frac{\mu n}{p^{0.9\de}}\r) 2^{m/2} \log p^n +O(2^m/p^n)$ where $\tau<1$ is a constant depending on $\de \times \mu, L,C$.
\end{cor}
\begin{proof}Choose $\kappa=\mu/10L$. By Lemma \ref{lem:gapzpn}, a $(\mu n, p^{\de})-GAP$ is $(\mu n/p^{0.9\de},0.1, \eps)$-additive of entropy rate $ \de \mu$. To use Theorem \ref{thm:extrstruczpn}, we need $\kappa L=0.1 \de \mu$ which is true by the choice of $\kappa$. Now choose $\gamma=\de \mu/CL$ for a large enough $C$. Then, for $X$ that is $\tau(\de \mu, C,L)n/L, |X|^{1-1/C})$-list decodable and $n \leq p^L-2$, the hypothesis of Theorem \ref{thm:extrstruczpn} is satisfied and hence the statement follows.
\end{proof}

We now state our corollary for Bohr sets. As in the previous section, we state it for constant $\rho$ and for $d=\mu n$ for simplicity.
\begin{cor}Let $C, \rho, \alpha, \mu>0$ be arbitrary. There exists $p_0, L_0 \in \N$ such that for all $L \geq L_0$ and primes $p \geq p_0$ the following holds. There exists an efficient $\eps$-extractor $\ext:\Z_p^n \rightarrow \{0,1\}^m$ for $(d=\mu n,\rho)$-Bohr sources in $\Z_p^n$ where $p\geq \max\{n^{1/(L-2)},\Omega\l(\l(\frac{n}{\alpha}\r)^{\mu}\r)\}$, $X$ is $(\tau n/L, |X|^{1-1/C})$-list decodable and $\eps<  \l(3\alpha+|X|^{-\tau}\r) 2^{m/2} \log p^n +O(2^m/p^n)$ where $\tau<1$ is an arbitrarily small constant depending on $d,\rho$ and $C$.
\end{cor}
\begin{proof}By Lemma \ref{lem:bohrG}, any $\Bo(S, \rho)$ is a $\l(100\kappa d, \beta, \eps  \r)$-additive source of entropy rate $\de$ in $\Z_p$ for $\kappa <1/100d$ whenever $p^n \geq \max {\l(\frac{4^{1/\eps}}{\rho}\r)^d, \l(\frac{1}{\rho}\r)^{d/1-\de}, \l(\frac{1}{\kappa \rho}\r)^{d/1-\beta}}$. Now apply Theorem \ref{thm:extrstruczpn} and using the lower bound on $p$ the conclusion follows.
\end{proof}

\subsection{Application to affine sources and a new XOR lemma}\label{sub:affine}

We note that extractor for additive sources in $\Z_p^n$ presented above indeed works for arbitrary affine spaces of constant min-entropy without any condition on list decodability as shown in Appendix \ref{app:affine}. Firstly we need a way of converting exponential sum bounds to extractors. This has been folklore and known as the Vazirani XOR lemma. However, the conditions required for that are too stringent for our character sum bounds and we need a different generalization of the XOR lemma which we state below.
\begin{lem}\label{lem:charsumzn}Let $M<N$ be integers with $M,N$ coprime and $N$ be the product of $n$ distinct primes all greater than $p$. Let $\sigma:\Z_N \rightarrow \Z_M$ be the function $\sigma(x)=x \mod M$. Let $X$ be a distribution on $\Z_N$ with $|\E_{X}\psi(X)| \leq \eps$ for every $\psi \in \Z_N^*$. Then, \[|\sigma(X)-U| = O\l(\l(\eps +n/p\r)\log N/M]\r)\] 
\end{lem}
The proof will perform a rather careful analysis of the traditional proof of the XOR lemma. See Appendix \ref{app:xor}. We believe this might be of independent interest.

\section{Extractor for APs and GAPs in $\F_q^n$}\label{sec:extrlines}

We first focus our attention to the special case of line sources. We construct an extractor for line sources and later generalize to partial lines (or $k$-lines). 

\subsection{Extractor for lines in $\Fq^n$}
As mentioned in the introduction, it becomes increasingly harder to construct an extractor for lines for small $q$ (large $n$), since when $n$ is large enough compared to $q$, we get a proof of non-existence by the density Hales-Jewett theorem. In this section, we shall focus on $1$-bit extractors. Generalizations to more number of bits follows from the XOR lemma (Lemma \ref{lem:addlem}). In the following, let $q$ be power of $p$.

For the sake of completeness, we first show by a simple well known probabilistic argument, the existence of a $1$-bit $0.1$-extractor for lines sources in $\F_q^n$ as long as $q=\Omega(n\log n)$.
\begin{lem}There exists a non-explicit $0.1$-extractor $f:\F_q^n \rightarrow \{0,1\}$ for all line sources in $\F_q^n$ for $n$ large enough as long as $q>200n\log n$.\end{lem}
\begin{proof}
Choose a random $f$ such that for each $x$, $\Pr[f(x)=0]=1/2$. Fix an arbitrary source $X=\{a+tb:t \in \F_q\}\subseteq \F_q^n$. Recall that we view a set as a source which is uniform on the set. Let $U$ denote the uniform distribution on $\{0,1\}$. For $i=0, \ldots q-1$, let $Y_i$'s be $0-1$ indicator random variables such that $Y_i=1$ iff $f(a+ib)=1$. We want to bound the event that $|f(X)-U|>0.1$. This is equivalent to the event $\left|\frac{1}{q}\sum_i Y_i-1/2 \right|>0.1$. Call the above event $E_X$. By a Chernoff bound, $\Pr[E_X]<2 \exp(-0.02q)$. By a union bound over all sources $X$, and noting that there are $q^{2n}$ lines, $\Pr[f\ \text{is not a 0.1 extractor}]<2 \exp(-0.02q)q^{2n}\leq 1$ by using the lower bound on $q$.
\end{proof}

Gabizon and Raz \cite{GabR} achieved an extractor for $q=\Omega(n^2)$.
\begin{thm}[\cite{GabR}]There is an explicit efficient $\eps$-extractor $\ext:\F_q^n \rightarrow \{0,1\}$ for all line sources in  $\F_q^n$ where $\eps \leq n/\sqrt{q}$. \end{thm}

In this section, we construct our extractor which beats even the randomness argument and works for $q=\Omega(n)$.

\paragraph{The Main Theorem}
We state our main theorem of this subsection. As in the previous sections, the theorem follows from a lemma on exponential sums and the XOR lemma.

\textbf{Theorem~\ref{thm:normlines}.} There is an explicit efficient $\eps$-extractor $\ext:\F_q^n \rightarrow \{0,1\}$ for all line sources in  $\F_q^n$ where $\eps \leq 4(n/q)^{1/2}$. 

In order to construct our extractor, we shall be using {\it Norm Polynomials} p.272 of \cite{LN}.
\begin{define}[Norm Polynomial]  A norm polynomial $P \in F_q[r_1, \ldots r_k]$ is a homogeneous polynomial of degree $k$ which satisfies for all $(c_1, c_2, \ldots , c_k) \in \F_q^k$,  $P(c_1,\ldots c_k)=0$ iff $c_1=\ldots c_k=0$.
\end{define}

\paragraph{Construction of Norm Polynomials}
We follow the construction given in \cite{LN}. Let $\alpha_1, \ldots \alpha_k$ be a basis of $E=F_{q^k}$ over $\F_q$. Set \[P(x_1, \ldots x_k)=\prod_{j=0}^{k-1}\left( \alpha_1^{q^j}x_1+\ldots +\alpha_k^{q^j}x_k\right)\]
Since, the $\alpha_{i}^{q^j}$, $j=0,1,\ldots k-1$, are conjugates of $\alpha_i$ with respect to $\F_q$, the coefficients of $N$ are in $\F_q$. Clearly, degree of $N$ is $d$. Now let $(c_1, \ldots c_k)\in \F_q^n$. Then,
\begin{eqnarray*}
P(c_1, \ldots c_k)&=&  \prod_{j=0}^{k-1}\left( \alpha_1^{q^j}c_1+\ldots +\alpha_k^{q^j}c_k\right)\\
&=&\prod_{j=0}^{k-1}\left( \alpha_1 c_1+\ldots +\alpha_k c_k\right)^{q^j}\\
&=&\left( \alpha_1 c_1+\ldots +\alpha_k c_k\right)^{\frac{q^k-1}{q-1}}
\end{eqnarray*}

which is zero if and only if $\alpha_1 c_1+\ldots +\alpha_k c_k=0$, which is true iff $c_i=0$ for all $1 \leq i \leq k$.

We now begin with the two main lemmas of this section. The first lemma is for additive characters and works for all $q$. The second lemma is for the quadratic multiplicative character for odd $q$.
\begin{lem}\label{lem:lines} There is an explicit efficient $f:\F_q^n \rightarrow \F_q$ such that the following holds. Let $X$ be a line in $\F_q^n$. Then for any non trivial additive character $\psi$, \[\frac{1}{q}|\sum_{x \in X}\psi(f(x))|\leq 4(n/q)^{1/2}\]
\end{lem}

\begin{lem}\label{lem:linesodd} Let $q$ be odd. There is an explicit efficient $f:\F_q^n \rightarrow \F_q$ such that the following holds. Let $X$ be a line in $\F_q^n$. Then for the multiplicative quadratic character $\chi_2$ we have
\[\frac{1}{q}|\sum_{x \in X}\chi_2(f(x))|\leq 4(n/q)^{1/2}\]
\end{lem}

To extract more bits, we use the XOR lemma along with the theorem on additive characters. The general form is presented in the next subsection. Let us now focus on the problem of extracting $1$ bit.
When $q$ is even, we use the trace function for the additive character which gives $1$ bit. For odd $q$, we see that the quadratic character outputs $1$ bit. Some care needs to be taken in this case. For a proof, see \cite{GabR}.

We start with the proof of Lemma \ref{lem:lines}.

\begin{proof}[Proof of Lemma \ref{lem:lines}]The construction is in two steps. First we define for an arbitrary subset of coordinates $S$, a polynomial $Q_S$ and then partition the $n$ coordinates carefully and apply a linear combination of the corresponding $Q_S$'s.
\paragraph{Construction of $Q_S$} Let $X=\{a+td:t \in \F_q\}$. For any subset of coordinates, say $S=\{x_1,\ldots ,x_k\}$, we define $Q_S(t)=P(a_1+td_1, \ldots a_k+td_k)$.
Now, observe that
\begin{itemize}
\item Coefficient of $t^k$ in $Q_S(t)$ is $P(d_1,\ldots d_k)$ which is zero iff $d_1=\ldots d_k=0$.
\item On the other hand, we also have that if $d_1=\ldots d_k=0$, then $deg(Q_S)=0$.
\end{itemize}

\paragraph{Combining the $Q_S$'s}
Now, the construction is as follows. We partition the $n$ coordinates into blocks of length $1,2,3,\ldots ,d$($d \leq 2\sqrt{n(1+1/p)}$) excluding multiples of $p$. Without loss of generality, we can assume that $n$ is exactly partitioned in the increasing order as mentioned above. If not, we can always append all-zero coordinates and work in a dimension $<2n$. (We adjust for this extra factor in the end. For now we assume $n$ can be exactly partitioned) Let us call this family of subsets of coordinates $\mathbb{S}$. We let $f(t)=\sum_{S \in \mathbb{S}}Q_S(t)$. (We abuse notation and sometimes use $f(t)$ and $f(x)$ interchangeably with the obvious correspondence.) We now argue that this polynomial is nonzero whenever some $d_i$ is nonzero. Now starting from the rightmost coordinate, we stop when we hit a nonzero $d_i$. All the blocks to its right will have degree $0$ and all the ones to the left will have degree less than the degree of this block. So there is no cancellation. Thus, we always have a non zero polynomial of degree $d \leq 2\sqrt{2n(1+1/p)}<4 \sqrt{n}$ (taking the extra doubling of dimension into account) and we can apply Theorem \ref{thm:weiladd} noting by the choice of the partition that $gcd(q,d)=1$ (as $d$ is never a multiple of $p$) to get \[\frac{1}{q}|\sum_{t \in \F_q}\psi(f(t))|\leq 4(n/q)^{1/2}\]
\end{proof}

Next, we prove Lemma \ref{lem:linesodd}.

\begin{proof}[Proof of Lemma \ref{lem:linesodd}]The construction is again in two steps. The first part is like in the previous proof but we state it for completeness. First we define for an arbitrary subset of coordinates $S$, a polynomial $Q_S$ and then partition the $n$ coordinates carefully and apply a linear combination of the corresponding $Q_S$'s.
\paragraph{Construction of $Q_S$} Let $X=\{a+td:t \in \F_q\}$. For any subset of coordinates, say $S=\{x_1,\ldots ,x_k\}$, we define $Q_S(t)=P(a_1+td_1, \ldots a_k+td_k)$.
Now, observe that
\begin{itemize}
\item Coefficient of $t^k$ in $Q_S(t)$ is $P(d_1,\ldots d_k)$ which is zero iff $d_1=\ldots d_k=0$.
\item On the other hand, we also have that if $d_1=\ldots d_k=0$, then $deg(Q_S)=0$.
\end{itemize}

\paragraph{Combining the $Q_S$'s}
Now, the construction is as follows. We partition the $n$ coordinates into blocks of length $1,3,\ldots ,d$($d \leq 2\sqrt{n}$), that is, excluding multiples of $2$. Without loss of generality, we can assume that $n$ is exactly partitioned in the increasing order as mentioned above. If not, we can always append all-zero coordinates and work in a dimension $<2n$. (We adjust for this extra factor in the end. For now we assume $n$ can be exactly partitioned) Let us call this family of subsets of coordinates $\mathbb{S}$. We let $f(t)=\sum_{S \in \mathbb{S}}Q_S(t)$. (We abuse notation and sometimes use $f(t)$ and $f(x)$ interchangeably with the obvious correspondence.) We now argue that this polynomial is nonzero whenever some $d_i$ is nonzero. Now starting from the rightmost coordinate, we stop when we hit a nonzero $d_i$. All the blocks to its right will have degree $0$ and all the ones to the left will have degree less than the degree of this block. So there is no cancellation. Thus, we always have a non zero polynomial of odd degree $d \leq 4 \sqrt{n}$ (taking the extra doubling of dimension into account) and we can apply Theorem \ref{thm:weilmult} noting by the choice of the partition that the polynomial can never be a perfect square since it is of odd degree to get \[\frac{1}{q}|\sum_{t \in \F_q}\chi_2(f(t))|\leq 4(n/q)^{1/2}\]
\end{proof}

\subsection{Extractors for APs and GAPs in  $\Z_{p}^n$}\label{sec:extrgap}
We will build on the polynomial obtained in the previous subsection to get an extractor for APs and GAPs. As we will use field operations, it will be convenient to use the
notation $\Fn$ rather than $\Z_{p}^n$.
Fix integers $r,s$ with $1\leq s \leq p-1$.
For $a_1,\ldots,a_r,b \in \Fn$ we denote by $G_{a_1,\ldots,a_r,b}$
the $(r,s)$-GAP $G_{a_1,\ldots,a_r,b} \triangleq \set{ \sum_{i=1}^r a_i\cdot t_i + b:0 \leq t_i \leq s-1 }$.

\noindent It will be convenient to look at GAPs where the $a_i$'s are linearly independent.
\begin{define}\label{dfn:indGAP}
For $a_1,\ldots,a_r,b\in Z_p^n$,
we say the $(r,s)$-GAP  $G_{a_1,\ldots,a_r,b}$ is
\emph{independent} if $a_1,\ldots,a_r$ are linearly independent in $\Z_p^n$.
\end{define}

\begin{claim}\label{clm:GAPcontainsInd}
An $(r,s)$-GAP in $\Fn$
can be written as a union of $(k,s)$-GAPs in $\Fn$ which are independent,
for $k \geq  r\cdot \log s /\log p$.
\end{claim}
\begin{proof}
Fix an $(r,s)$-GAP  $G_{a_1,\ldots,a_r,b}$ .
Let $k$ be the dimension of the $\Fp$-linear span
of $\set{a_1,\ldots,a_r}$.
We have
\[p^{k}\geq s^r \rightarrow k \geq r\cdot \log s /\log p.\]
Assume, w.l.o.g., that $a_1\ldots,a_{k}$ are linearly independent.
We can write $G_{a_1,\ldots,a_r,b}$ as a union of GAPs $G_{a_1,\ldots,a_k,b'}$
where $b'$ will range over the values $\set{a_{k+1}\cdot t_{k+1} + \ldots + a_r\cdot t_r + b:0 \leq t_i \leq s-1}$.
\end{proof}

\cite{GabR} and  \cite{deVG} used polynomials that are non-constant over subspaces
of a certain dimension together with Weil bounds to construct affine extractors.
We show that such polynomials are sufficient for the more general goal of constructing extractors for GAPs.
For a polynomial $f:\Fn\to \Fp$, and $a_1,\ldots,a_k,b \in \Fn$,
we denote by $f|_{a_1,\ldots,a_k,b}$ the polynomial
$f|_{a_1,\ldots,a_k,b}(t_1,\ldots,t_k)\triangleq f(a_1\cdot t_1+\ldots +a_k\cdot t_k + b)$.
We first prove the following theorem.
\begin{thm}\label{thm:poly-to-GAPext}
Fix integers $r,s,d$ with $d,s<p$.
Fix integer $k\leq  \min\set{1,r\cdot \log s /\log p} $.
Suppose we are given an efficiently computable polynomial $f:\Fn \to \Fp$ such that for all $a_1,\ldots,a_k,b \in \Fn$,
where $a_1,\ldots, a_k$ are linearly independent,
$f|_{a_1,\ldots,a_k,b}$ is non-constant of degree $1<d'\leq d$.

Then we can construct an explicit efficient $\eps$-extractor $\ext:\Fn\to \B^m$ for $(r,s)$-GAP sources where
\[\eps \leq (4 \log p\cdot \sqrt p +1)\cdot d/s\cdot  2^{m/2}+ 2^m/p.\]

\end{thm}
\begin{remark}
In the case $r=1$ a $d/s$ factor can be taken off from the $\eps$. This
will be evident in the proof.
\end{remark}

Once we have this, the two main theorems follow immediately.\\ \\ 
\textbf{Theorem~\ref{THM:normap}.} There is an explicit efficient $\eps$-extractor $\ext:\F_p^n \rightarrow \{0,1\}^m$ for all $k$-AP sources in  $\F_p^n$ where $\eps \leq 16 \log^2 p \sqrt{np} 2^{m/2}/k$.

\begin{proof}
Plugging the polynomial from Lemma \ref{lem:lines}
that has degree $4\sqrt n$ and is non-constant on affine subspaces of dimension $1$ completes the proof. The $2^m/p$ factor can be dropped as it will be dominated by the first term.
\end{proof}

\begin{thm}[Extractors for GAPs]\label{THM:normgap}\label{thm:ext_for_dim_k}
Fix integers $r,s$ with $s<p$.

Then we can construct an explicit efficient $\eps$-extractor $\ext:\Fn\to \B^m$ for $(r,s)$-GAP sources where
\[\eps \leq (34 \log^3 p\cdot \sqrt p )\cdot n/(r\cdot \log s \cdot s)\cdot  2^{m/2}+ 2^m/p.\]
In particular, when $p = \Omega( (r\cdot s /n)^2)$ we can output one bit with constant error.

\end{thm}

\begin{proof}DeVos and Gabizon (\cite{deVG}, Theorem $7$) construct an explicit function
$f:\Fn\to \Fp$ that is non-constant of degree $1<d<2n/k$
when restricted to any affine subspace of dimension $k$.
Plugging this function into Theorem \ref{thm:poly-to-GAPext} finishes the proof.
\end{proof}

Let us now prove Theorem~\ref{thm:poly-to-GAPext}. The main ingredient in the theorem's proof is the following lemma that generalizes the Weil bound for exponential sums (Theorem \ref{thm:weiladd}) to the case where the sum ranges only over an AP, rather than the whole field.
\begin{lem}\label{lem:GAPexpsum}
Let $f\in \Fp[t]$ be a polynomial of degree $1<d<p$. Let $X$ be an $s$-AP.
Let $\psi$ be a non trivial additive character of $\Fp$.
Then, for any integer $0<s\leq p$,
\[ \left|\sum_{t \in X}\psi(f(t)) \right|\leq 4 \log p\cdot \sqrt p\cdot d.\]
\end{lem}
Using the XOR lemma, Lemma \ref{lem:GAPexpsum} implies the following.
\begin{cor}\label{cor:GAPextfrompoly}
Let $f\in \Fp[t]$ be a polynomial of degree $1<d<p$. For any integer $0<s\leq p$, let $X$ be an $s$-AP source.
Let $\sigma:\Fp^n\to \Fp^m$ be the function from Lemma \ref{lem:addlemZpn}.
Then,
$|\sigma(X)-U|<\eps$
for $\eps = 4 \log p\cdot \sqrt p\cdot d\cdot p^{m/2}$.
\end{cor}

We prove Theorem \ref{thm:poly-to-GAPext} given the corollary.
\begin{proof}

Let $\sigma:\Fp^n\to \Fp^m$ be the function from Lemma \ref{lem:addlemZpn}.
Fix an $(s,k)$-independent GAP $X=G_{a_1,\ldots,a_k,b}$.

Claim \ref{clm:GAPcontainsInd} implies it is enough to construct an extractor for
such GAPs.

We know that $g\triangleq f|_{a_1,\ldots,a_k,b}$
is non-constant of degree $d'$ where $1< d'\leq d<p$.
Assume w.l.o.g. that $t_1$ appears in $g$ with degree greater than $1$.
Denote by $a$ the maximal degree that $t_1$ has in $g$.
Write $g$ as a polynomial in $t_1$ whose coefficients are polynomials in $t_2,\ldots, t_k$.
Look at the coefficient $g_a(t_2,\ldots,t_k)$ of $t_1^a$ in $g$.
From the Schwartz-Zippel Lemma $g_a(t_2,\ldots,t_k) = 0$ with probability
at most $d'/s\leq d/s$ when choosing $t_2,\ldots,t_k$ uniformly in $\set{0,\ldots,s-1}^{k-1}$.
Note that $X$ can be viewed as a convex combination of the
$s$-APs $X_{t_2,\ldots,t_k}\triangleq \set{a_1t_1+a_2t_2+ \ldots + a_kt_k+b|0\leq t_1\leq s-1}$ .
For $t_2,\ldots, t_k$  such that $g_a(t_2,\ldots,t_k) \neq 0$, it follows that
the polynomial $g_{t_2,\ldots,t_k}(t)\triangleq g(t,t_2,\ldots,t_k)$ is non-constant of degree
larger than 1 and at most $d$. Therefore, from Corollary \ref{cor:GAPextfrompoly} we have $|\sigma(X_{t_2,\ldots,t_k})-U| \leq 4 \log p\cdot \sqrt p\cdot d\cdot p^{m/2}$.
\end{proof}

\noindent
We proceed with the proof of Lemma \ref{lem:GAPexpsum}.
\begin{proof} (of Lemma \ref{lem:GAPexpsum})

The proof combines the Weil bound with Fourier analysis.
It is based on two claims.
The first uses the Weil bound to bound the Fourier coefficients of
$f$ \emph{composed with an additive character}.

\begin{claim}\label{clm:max}
Let $\psi$ be the non trivial additive character from the lemma statement. Then for all $\xi \in \F_p$, we have $\left| \widehat{\psi \circ f }(\xi) \right| \leq \sqrt{d/p}$.
\end{claim}
\begin{proof}
Suppose $\psi(x) = e_p(a\cdot x)$ for some $a\in \Fp$.
We have
\begin{eqnarray*}
&&\left| \widehat{\psi \circ f} (\xi) \right|\\
&=&1/p \left| \sum_{t \in \F_p}e_p(a\cdot f(t)) \cdot e_p(-\xi t))\right|\\
&=&1/p \left| \sum_{t \in \F_p}e_p(a^{-1}(f(t)-a\cdot \xi \cdot t))\right|\\
&\leq & (d/p)^{1/2}
\end{eqnarray*}
The last line follows from Weil bound and by observing two things: The first is that the
sum in the line before is an exponential sum with the character $\phi'(x) = e_p(a^{-1}\cdot x)$ on the polynomial $f'(t)\triangleq f(t)-\xi \psi ^{-1} t$.
The second is that $f'(t)$ is also a non-constant polynomial of degree $d<p$ so the Weil bound can be used.
\end{proof}

Next, we need the following claim upper bounding the $L_1$ Fourier norm of a set related to $s$-APs.
Let $A=\{0,1,\ldots s-1\}$. Denote by $A(x)$ the indicator set of $A$.
\begin{claim}\label{clm:l1}$\sum_{0 \leq j \leq p-1}\left|\hat{A}(j)\right| \leq 4 \log p$\end{claim}
\begin{proof}
Note that
\begin{eqnarray*}
\hat{A}(j)&=&1/p \sum_{i \in A}e(ji)\\
&=&1/p \frac{e(jk)-1}{e(j)-1}
\end{eqnarray*}

Now noting that $|e(\theta)-1| \geq 4 \{\theta \}$
we have $\left| \hat{A}(j) \right| \leq \frac{1}{2p \{j/p\}}$

We now turn to computing the $L_1$ Fourier norm of $A$.
\begin{eqnarray*}
\sum_{j \in \F_p }\left| \hat{A}(j) \right|&=&\sum_{j \leq p/2 }\left| \hat{A}(j) \right|+\sum_{j > p/2}\left| \hat{A}(j) \right|\\
& \leq & \sum_{j \leq p/2 } 1/(2j)+ \sum_{j > p/2} 1/2(p-j)\\
& \leq & 4 \log p
\end{eqnarray*}
\end{proof}

With the above two claims in place, we now turn to proving the lemma.
\begin{eqnarray*}
\left|\sum_{0\leq t\leq s-1}\psi(f(t))\right|&=&\left| \sum_{t \in \F_p}A(t) \psi(f(t))\right|\\
&=&p\left| \sum_{\xi \in \F_p}\overline{\widehat{A}(\xi)} \widehat{\psi \circ f} (\xi)\right|\\
& \leq & 4 \log p \cdot \sqrt p \cdot d.
\end{eqnarray*}

This finishes the proof.
\end{proof}

\section{Acknowledgements}We thank the anonymous referees for their valuable comments to improve the quality of the writeup. We also thank the reviewer who pointed out a problem with the condition of list decodability in Section~\ref{sec:extrzpn}. This led us to making the condition much less restrictive.
\bibliographystyle{alpha}
\newcommand{\etalchar}[1]{$^{#1}$}

\appendix

\section{Extracting more randomness from GAPs in $\Z_p$ under Paley graph conjecture}\label{app:pgc}
We first state the conjecture below and then the main theorem of this section.

\begin{conjecture}[$PG(\de)$]Let $\chi$ be a multiplicative character of $\Z_p^*$. Let $S,T \subseteq \Z_p$ such that $|S|, |T| >p^{\de}$. Then there exists $\eps>0$ such that \[\l| \sum_{s \in S, t \in T}\chi(s+t)\r|<|S||T|/p^{\eps}\]
\end{conjecture}

We begin with the main theorem of this section.
\begin{thm}\label{thm:PG}Let $p$ be such that $p-1$ have no large prime factors. There exists an explicit efficient  $2p^{-\eps}2^{m/2}$-extractor $\ext:\Z_p \rightarrow \{0,1\}^m$, for $\l(r, p^{4\de /r}\r)$-GAP sources (entropy rate $4 \de$) under $PG(\de)$ where $\eps=\eps(\de)$.
\end{thm}
We first prove a statement about APs ($r=1$) and later a statement for $r>1$ and later use Lemma \ref{lem:multlem} to prove the main theorem .

\begin{lem}Given an $AP$ $X=\{b_0+ab_1: 0 \leq a <l\}$, for any $k$ such that $l>p^{2\de}$, there is $\eps>0$ such that we have \[\l|\E_{x \in X}\chi(x)\r|<1/p^{\eps}+1/p^{\de}\]
\end{lem}
\begin{proof}
Let $k$ be an arbitrary integer and $l=kd$ such that $d>p^{\de}$. Remove this assumption later. 
For $-1 \leq i \leq k$, let \[S_i=\{ab_1: id \leq a <(i+1)d\}\] Also, let \[T=\{b_0+ab_1: 0 \leq a <d\}\] Now, $|S_i|, |T| = d>p^{\de}$ Thus, applying Conjecture $PG(\de)$, to each $S_i$ and $T$ pair, we get \[\l|\sum_{s \in S_i, t \in T}\chi(s+t)\r|<d^{2}/p^{\eps}\]
Consider $A=\sum_i \sum_{s \in S_i, t \in T}\chi(s+t)$. By the above, $|A|< (k+2)d^2/p^{\eps}$. Rewriting $A$ in another way, \[A=\sum_z m(z)\chi(z)\] where $m(z)$ is the number of times $\xi(z)$ occurs in $A$. Consider any $z=b_0+(xd+yb)$ where $0 \leq x \leq k-1$. Then $m(z)$ includes contribution from the following $2$ sets: $S_i$, $i \in \{x-1,x\}$.
Therefore, 
\begin{eqnarray*}
m(z)&= &\sum_{a \in \{d-y,y\}}a\\
&=&(d-y+y)\\
&=&d
\end{eqnarray*}
Now, if some $x_i=0$ or $x_i=k$, we won't have the above cancellation. We have $2(k-(k-1))d=2d$ such $z's$ each with $m(z) \leq d$.
Therefore, $|d\sum_{x \in X}\chi(x)|\leq kd^{2}/p^{\eps}+2d^{2}$. Setting $k=p^{\de}$ proves the theorem.

\end{proof}


\begin{lem}Given an $(r,l)$-GAP $X=\{b_0+\sum_{i=1}^ra_ib_i: 0 \leq a_i <l\}$, for $r \geq 2$, $l>p^{\frac{\de}{\lfloor r/2 \rfloor}}$, we have \[\l|\E_{x \in X}\chi(x)\r|<1/p^{\eps}\]
\end{lem}
\begin{proof}
Since the GAP is proper, let $m=\lfloor r/2 \rfloor$. Let \[S=\{b_0+\sum_{i=1}^{m}a_ib_i: 0 \leq a_i <l\}\] and \[T=\{\sum_{i=m+1}^{r}a_ib_i: 0 \leq a_i <l\}\]
Then, since $|S|$, $|T|>p^{\de}$, by Conjecture $PG(\de)$, we have \[\l|\sum_{s \in S, t \in T}\chi(s+t)\r|=\l|\sum_{x \in X}\chi(x)\r|<1/p^{\eps}\]
This proves the theorem.
\end{proof}

Thus, combining the two lemmas above, we have the following.
\begin{lem}Given an $(r,l)$-GAP $X=\{b_0+\sum_{i=1}^ra_ib_i: 0 \leq a_i <l\}$, for $r \geq 1$, $l>p^{\frac{4\de}{r}}$, we have \[|\E_{x \in X}\chi(x)|<2/p^{\eps}\]
\end{lem}

The proof of the main theorem now follows using the above lemma and the XOR lemma (Lemma \ref{lem:addlem}) and the efficiency of the extractor follows from the fact that the discrete logarithm in $\Z_p$ is efficiently computable if $p-1$ is smooth.

\section{Extractors from character sum bounds}\label{app:xor}
We generalize the XOR lemma on $\Z_N$ under a relaxed requirement on the character sums as made precise below. We note that not all the statements are specific to $\Z_N$. However, we stick to it for simplicity.
Throughout, let $N=\prod_{i=1}^nq_i$ where each distinct $q_i \geq p$ for some prime $p$. Let $\Z_N^*$ denote the units of $(\Z_N,*)$.
We restate our main lemma below. The proof will be similar to that in \cite{Rao} but will require a more careful analysis.

\noindent \textbf{Lemma \ref{lem:charsumzn}} Let $M<N$ be integers with $M,N$ coprime and $N$ be the product of $n$ distinct primes all greater than $p$. Let $\sigma:\Z_N \rightarrow \Z_M$ be the function $\sigma(x)=x \mod M$. Let $X$ be a distribution on $\Z_N$ with $|\E_{X}\psi(X)| \leq \eps$ for every $\psi \in \Z_N^*$. Then, \[|\sigma(X)-U| = O\l(\l(\eps +n/p\r)\log N/M]\r)\]

Before we prove our lemma, we need the following.

\begin{claim}[Prop 2.9 in \cite{Rao}]\label{clm:cs}$\sum_{x \in \Z_N}|f(x)| \leq N^{3/2}\max_{\xi \in \Z_N}\l|\hat{f}(\xi)\r|$
\end{claim}

\begin{lem}[Lemma 4.4 in \cite{Rao}]\label{lem:rao}Let $M<N$ be integers. Let $\sigma:\Z_N \rightarrow \Z_M$ be the function $\sigma(x)=x \mod M$. Then, for every character $\phi \in \Z_M$, we have \[\sum_{\xi \in \Z_N}\l|\widehat{\phi \circ \sigma} (\xi)\r| = O(\log N)\] 
\end{lem}

The following claim is implicit in Lemma 4.4 in \cite{Rao}.
\begin{claim}\cite{Rao}\label{clm:hp}Let $M<N$ be integers and $w \in \Z_M$. Then \[\sum_{\xi \in \Z_N, \xi \neq wN/M}\l|\frac{1}{e(\frac{\xi M - w N}{MN})-1}\r| \leq O( N\log N)\]
\end{claim}

We will also need the following lemmas.
\begin{lem}\label{lem:xor}Let $X$ be a distribution on $\Z_N$ with $|\E_{X}\psi(X)| \leq \eps$ for every $\psi \in \Z_N^*$. Let $U$ be the uniform distribution on $\Z_N$. Let $\sigma: \Z_N \rightarrow \Z_M$ be a function such that for every non trivial character $\phi \in \Z_M$, we have \[\sum_{\xi \notin \Z_N^*}\l|\widehat{\phi \circ \sigma}(\xi)\r| \leq \tau_1\]
and \[\sum_{\xi \in \Z_N}\l|\widehat{\phi \circ \sigma}(\xi)\r| \leq \tau_2 .\]
Then, $|\sigma(X)-\sigma(U)|< \l(\eps \tau_2 +\tau_1\r) \sqrt{M}$.
\end{lem}
\begin{proof}Let $\phi \in \Z_M$ be a nontrivial character. Note that the hypothesis $|\E_{X}\psi(X)| \leq \eps$ is equivalent to $|\widehat{X-U}(\psi)| \leq \eps /N$ for $\psi \neq 0$.
Then,
\begin{eqnarray*}
&&\l|\widehat{\sigma(X)-\sigma(U)}(\phi)\r|\\
&=&1/M\l|\sum_{x \in \Z_N}\phi \circ \sigma (x) (X-U)(x)\r|\\
&=&N/M\l|\sum_{\xi \in \Z_N}\widehat{\phi \circ \sigma} (\xi) \widehat{X-U}(\xi)\r|\\
&=&N/M\l|\sum_{\xi \in \Z_N^*}\widehat{\phi \circ \sigma} (\xi) \widehat{X-U}(\xi)\r|\\
&&+N/M\l|\sum_{\xi \notin \Z_N^*}\widehat{\phi \circ \sigma} (\xi) \widehat{X-U}(\xi)\r|\\
&\leq  & \eps \tau_2/M+\tau_1 /M
\end{eqnarray*}

Now, for the trivial character we have $\widehat{\sigma(X)-\sigma(U)}(\phi)=0$ since $X$ and $U$ are distributions.
Thus, the lemma follows by Claim \ref{clm:cs}.
\end{proof}

Next, we prove the following lemma.
\begin{lem}\label{lem:genrao}Let $M<N$ be integers with $M,N$ coprime and $N$ be the product of $n$ distinct primes all greater than $p$. Let $\sigma:\Z_N \rightarrow \Z_M$ be the function $\sigma(x)=x \mod M$. Then, for every non trivial character $\phi \in \Z_M$, we have \[\sum_{\xi \notin \Z_N^*}\l|\widehat{\phi \circ \sigma} (\xi)\r| \leq (Cn/p)\log N\] 
\end{lem}
\begin{proof}Recall that $e(x)=e^{2 \pi ix}$. Let $\phi$ be any non trivial character of $\Z_M$. Then $\phi(y)=e(wy/M)$ for some $w \in \Z_M - \{0\}$. Then, $\phi(\sigma(X))=e(wx/M)$.
Now,
\begin{eqnarray*}
&&\sum_{\xi \notin \Z_N^*}\l|\widehat{\phi \circ \sigma} (\xi)\r|\\
&=&1/N\sum_{\xi \notin \Z_N^*}\l|\sum_{x \in \Z_N}e(wx/M-\xi x/N)\r|\\
& = &1/N\sum_{\xi \notin \Z_N^*, \xi \neq wN/M}\l|\sum_{x \in \Z_N}e(wx/M-\xi x/N)\r| \text{  (as $w \neq 0$ and $N,M$ coprime)}\\
& = &1/N\sum_{\xi \notin \Z_N^*, \xi \neq wN/M}\l|\frac{e\l(N(\xi M-wN)/NM\r)-1}{e\l((\xi M-wN)/NM\r)-1}\r| \text{  (by geometric summation)}\\
& \leq &1/N\sum_{\xi \notin \Z_N^*, \xi \neq wN/M}\l|\frac{2}{e\l((\xi M-wN)/NM\r)-1}\r|\\
& \leq &1/N\sum_{prime \ q|N}\sum_{\xi' \in \Z_{N/q}, \xi' \neq wN/qM}\l|\frac{2}{e\l((\xi' M-w(N/q))/(N/q)M\r)-1}\r|\\
& \leq &1/N\sum_{prime \ q|N}2C(N/q)\log (N/q)\ \text{  (by Claim \ref{clm:hp})}\\
& \leq &(2Cn/p)\log N \text{  (every $q \geq p$)}
\end{eqnarray*}
\end{proof}

We now finish the proof of Lemma \ref{lem:charsumzn}.
\begin{proof}[Proof of Lemma \ref{lem:charsumzn}]  Using Lemma \ref{lem:genrao} and Lemma \ref{lem:xor}, we have \[|\sigma(X)-\sigma(U)|< \l(\eps \log N +n/p \log N\r) \sqrt{M}+2M/N\] On uniform input, the distribution $\sigma(U)$ is close to uniform. More precisely, let $N=sM+t$ with $t<M$. Therefore, $\sigma(U)$ is $2t((s+1)/N-1/M)=2M/N$ close to uniform. This finishes the proof.
\end{proof}

\section{Extractors for affine sources}\label{app:affine}
We note that extractor for structured sources in $\Z_p^n$ presented in Section \ref{sec:extrzpn} indeed works for arbitrary vector spaces of constant min-entropy without any condition on list decodability.

\begin{thm}\label{thm:extrstruczpnaf}There exists $p_0, L_0 \in \mathbb{N}$ such that for all $L \geq L_0$ and primes $p \geq p_0$ the following is true. Let  $\de>0$ be arbitrary. There exists an efficient $\eps$-extractor $\ext:\Z_p^n \rightarrow \{0,1\}^m$ for $(\de)$-affine sources (entropy rate $\de$) in $\Z_p^n$ where $n \leq p^{L-2}$ and $\eps<  \l(1/p^{\tau n}+n/p \r) 2^{m/2} \log p^n +O(2^m/p^n)$ where $\tau$ is a constant depending on $\de / L$.
\end{thm}

The proof will again follow from the following lemma on exponential sums and Lemma \ref{lem:charsumzn}.

\begin{lem}There exists $p_0, L_0 \in \mathbb{N}$ such that for all $L \geq L_0$ and primes $p \geq p_0$ the following is true. There exists an efficient $f:\Z_p^n \rightarrow \Z_q$ (for $p^n < q < p^{Ln}$) such that if $\de>0$ is arbitrary, then there exists $\tau(\de / L)>0$ such that if
\begin{itemize}
\item $X$ is a $\de$-affine source in $\Z_p^n$
\item $n \leq p^{L-2}$
\end{itemize}
Then, for all $\xi \in \Z_q^*$,
\[\l|\sum_{x\in X}e_q(\xi f(x))\r|<q^{-\tau}  |X|  \]
\end{lem}
The proof is along the lines of the general theorem above. We restate it here for completeness.
\begin{proof}Let $q_1, q_2, \cdots q_n$ be $n$ distinct primes such that for all $i$, $q_i=1\ (mod\ p)$.
This is guaranteed by the following claim which uses Linnik's theorem. 
\begin{claim}There exist $L_0>0$ and $n_0 \in \N$ such that  for all $p \geq n_0, L \geq L_0$, the size of the set \[\{q:q=1 \pmod p, q \leq p^L \}\] is at least $p^{L-2}$.
\end{claim}
\begin{proof}Define \[\theta(x,p)=\sum_{q=1 \pmod p, q \ prime,q \leq x}\log q\]
By Linnik's theorem, for all $p \geq n_0$, if $x>p^L$, we have $\theta(x,p) \geq \frac{Cx}{p^{1/2}\phi(p)}$ where $\phi(n)$ is the Euler $\phi$ function which counts the number of positive integers up to $n$. Using $\phi(p) \leq p$ we have $\theta(x,p) \geq \frac{Cx}{p^{3/2}}$. Let \[\pi(x,p)=\sum_{q=1 \pmod p, q \ prime,q \leq x}1\] 
Then, \[\pi(x,p) \log x \geq \theta(x,p) \geq  \frac{Cx}{p^{3/2}}\]
Thus,  $\pi(x,p) \geq  \frac{Cx}{p^{3/2}\log x}$ which we want to be at least $n$. Choose $x=p^{L'}$ for $L' \geq L$. Then $n \leq p^{L'-2}$ which finishes the proof.
\end{proof}
Thus, by the above, we have for all $i$, $q_i<p^{L}$. Also, let $g_i$ generate the order $p$ subgroup in $\Z_{q_i}^*$. Define two maps $\phi_1, \phi_2$ as follows. 
Let $\phi_1:\Z_p^n \rightarrow \prod_{i \in [n]}\Z_{q_i}$ be defined by \[\phi_1(x_1,x_2,\cdots x_{n})= (g_1^{x_1},\cdots g_{n}^{x_{n}})\]
and for $q=\prod_{i \in [n]}q_i$, let $\phi_2:\prod_{i \in [n]}\Z_{q_i} \rightarrow \Z_q$ be defined by \[\phi_2(y_1,\ldots y_n)=\sum_{i=1}^n  y_i \frac{q}{q_i}\l[\l(\frac{q}{q_i}\r)^{-1}\r]_{q_i} \in \Z_q\]
where $[x^{-1}]_p$ is the inverse of $x$ in $\Z_p^*$. Note that $\phi_2$ is the Chinese remaindering map.

Define function $f$ as follows. \[f:\Z_p^n \rightarrow \Z_q\]\[x \mapsto \phi_2 \circ \phi_1 (x)\] Let $Y=f(X)$. $|Y| \geq q^{\de/L}$.

Since $f$ is an one-one function from $\l(\Z_p^n,+\r)$ into $(\Z_q^*,*)$ we have that $Y$ is a multiplicative subgroup in $\Z_q^*$.

We now appeal to the following exponential sum.
\begin{thm}\cite{Bou:subgroup}\label{thm:bourzpn1}Given $\de>0$, there is $\tau>0$ such that the following holds. Let $q$ be an arbitrary modulus and $H \subseteq \Z_q^*$ be a multiplicative subgroup. Then, for all $\xi \in \Z_q^*$, $\l|\sum_{x\in H}e_q(\xi f(x))\r|<q^{-\tau} |H|$.\end{thm}

Using Theorem \ref{thm:bourzpn1}, we have for $\xi \in \Z_q^*$, \[\l|\sum_{x\in X}e_q(\xi f(x))\r|<q^{-\tau} |X|<1/p^{\tau n}|X|\]
\end{proof}

\end{document}